\pdfoutput=1
\documentclass[
superscriptaddress,
nofootinbib,longbibliography,
onecolumn,
notitlepage,
aps]
{revtex4-2}
\usepackage{amsmath,amsfonts,amssymb,amsthm,amsbsy,mathtools}
\usepackage{graphicx}
\usepackage{dcolumn}
\usepackage{bm}
\usepackage{physics} 
\usepackage{bbold}
\usepackage{mathtools}
\usepackage{xcolor}
\usepackage{tikz}
\usepackage{tikz-cd}
\usepackage{comment}
\usepackage{subfigure}
\usepackage{comment}
\usepackage{soul}
\usepackage{framed}
\usepackage{mdframed}
\usepackage{csquotes}
\usepackage{appendix}
\usepackage{hyperref}
\newtheorem{proposition}{Proposition}
\newtheorem{theorem}{Theorem}

\newtheorem{example}{Example}
\usepackage{natbib}
\usepackage[normalem]{ulem}
\usepackage{xcolor}
\usepackage{xspace}
\usepackage{ifthen}

\def\phys{{\mathrm{phys}}}
\def\kin{{\mathrm{kin}}}

\def\op{{\mathrm{op}}}
\def\PN{{\mathrm{PN}}}
\def\persp{{\mathrm{persp}}}

\def\S{{\mathbf{S}}}
\def\bbR{{\mathbb{R}}}
\def\H{{\mathcal{H}}}

\newcommand{\showcomments}{true}

\newcommand{\acd}[1]%
{\ifthenelse{\equal{\showcomments}{true}}{{\color{magenta}{#1}}}{\xspace}}%

\newcommand{\vk}[1]%
{\ifthenelse{\equal{\showcomments}{true}}{{\color{orange}{#1}}}{\xspace}}%

\newcommand{\caslavcom}[1]

\begin{document}

\author{Anne-Catherine de la Hamette}
\thanks{These authors contributed equally to this work.}
\affiliation{University of Vienna, Faculty of Physics, Vienna Doctoral School in Physics, and Vienna Center for Quantum Science and Technology (VCQ), Boltzmanngasse 5, A-1090 Vienna, Austria}
\affiliation{Institute for Quantum Optics and Quantum Information (IQOQI),
Austrian Academy of Sciences, Boltzmanngasse 3, A-1090 Vienna, Austria}
\affiliation{Institute for Theoretical Studies, ETH Zürich, 8006 Zürich, Switzerland}

\author{Viktoria Kabel}
\thanks{These authors contributed equally to this work.}
\affiliation{Institute for Theoretical Physics, ETH Zürich, 8093 Zürich, Switzerland}

\author{\v{C}aslav Brukner}
\affiliation{University of Vienna, Faculty of Physics, Vienna Doctoral School in Physics, and Vienna Center for Quantum Science and Technology (VCQ), Boltzmanngasse 5, A-1090 Vienna, Austria}
\affiliation{Institute for Quantum Optics and Quantum Information (IQOQI),
Austrian Academy of Sciences, Boltzmanngasse 3, A-1090 Vienna, Austria}

\title{Accessibility of Global Properties from Internal Quantum Reference Frame Perspectives}

\begin{abstract}
    A fundamental question in the field of quantum reference frames concerns what global properties of a system can be determined by observers operating entirely from within that system. We investigate this question by extending both the perspectival and perspective-neutral approaches beyond the commonly studied zero total momentum case to arbitrary fixed charge sectors. When the entire system, including the reference frames, moves at a fixed total momentum $P$ relative to an external frame, this global charge becomes encoded in the quantum states and transformations between reference frames. Our extension leads to modified relative states and observables with QRF transformations that induce an additional $P$-dependent phase, treating all charge sectors as equally fundamental. By granting the internal observers successively more access and resources, we identify under which conditions they can infer the total momentum. These results clarify the relationship between major QRF approaches -- perspectival, perspective-neutral, operational, and extra-particle -- showing how their differing conclusions stem from different assumptions about which observables are deemed accessible from within. Our findings cast light on the relation between local and global perspectives and raise fundamental questions about scenarios where no global perspective exists, contributing to a deeper understanding of relationality and the role of perspectives in quantum theory.
\end{abstract}

\maketitle

\section{Introduction}

Symmetries are ubiquitous in physics: from the translation-invariance of simple mechanical systems to the fundamental symmetries of modern particle physics, they pervade our physical theories. The presence of symmetries implies, or is implied by, the fact that only quantities that are invariant under the action of the symmetry group are observable. In presence of translation invariance, we can only measure relative positions, often we are further restricted to relative velocities by the Galilei group, and if our theory has a gauge symmetry, only gauge-invariant quantities are deemed physically meaningful. Yet, we often leave the relationality of our description implicit. It is common to simply speak of \enquote{the position} of a point particle in Newtonian mechanics, even if the Hamiltonian is translation-invariant. This is justified because of the implicit use of a reference system -- for instance, we might tacitly describe all positions relative to the wall of the laboratory. While these reference frames are usually not treated as physical degrees of freedom, in practice they are always instantiated by physical systems. Recognising this and taking into account that all material systems should ultimately be describable by quantum theory has led to the development of quantum reference frames (QRFs), a set of frameworks in which reference frames themselves are described quantum mechanically.

In the \emph{perspectival approach to QRFs}, which was proposed in \cite{Giacomini_2017_covariance} and further developed in \cite{Giacomini_2018_spin, delaHamette_2020, Mikusch_2021, Barbado_2020, Apadula_2022, Giacomini_2020_einstein, Giacomini_2021, Giacomini_2021_einstein, Cepollaro_2021, Kabel_2022_conformal, delahamette_2022_DiffsICO, Wang_2023, Kabel_2024, cepollaro2024}, the focus is on characterising the perspectives relative to different QRFs and the transformations between them. Its name signals a contrast with the \emph{perspective-neutral approach to QRFs}, in which the descriptions relative to different QRFs are derived from a common, perspective-neutral description \cite{Vanrietvelde_2018a, Vanrietvelde_2018b, Höhn_2018a, Hoehn_2018b, Hoehn_2019_trinity, Hoehn_2020, Krumm_2021, hoehn2021quantum, Hoehn_2021, delaHamette_2021_perspectiveneutral, delaHamette_2021_entanglement, Hoehn_2023_subsystems, Chataignier_2024, Araujoregado_2025}. This could be seen as the description relative to a (possibly fictitious) external reference frame, access to which is prohibited by the presence of symmetries. In the perspective-neutral approach, the restriction due to the symmetry is realised through the imposition of a constraint -- the vanishing of the total momentum, $P=0$, in the case of the translation group. Since the total momentum acts as a generator of translations on all systems, projecting onto the zero charge sector yields a space of invariant states: the physical Hilbert space. Further \enquote{jumping} into the perspective of a particular system by controlling on its position being in the origin recovers the same description relative to different QRFs and transformations between them as in the perspectival approach.

Finally, the \emph{operational} \cite{Loveridge_2017, Carette_2023} and the \emph{extra-particle approach} \cite{Castro_Ruiz_2021} allow for arbitrary charge sectors, though both approaches differ in whether to encode information about the global charge in an accessible subsystem or whether to average over it. More precisely, \cite{Castro_Ruiz_2021} retains access to global charges even at the perspectival level as part of the physical description, by encoding this information in an \enquote{extra particle}. The operational approach, on the other hand, averages over external structures so as to eliminate any trace of global information, though at the cost of losing invertibility of frame changes (cf.~\cite{Castro-Ruiz_2025}).\\

The perspective-neutral framework and its perspective-dependent variants for QRFs exemplify the familiar complementarity between local and global descriptions in physics. Given a collection of local descriptions, one may inquire about their compatibility with a global representation: does a global description exist, from which the local ones can be recovered, and is such a description unique? There are situations in which no single global counterpart exists -- for instance, in extended Wigner’s friend scenarios \cite{Brukner_2017, Frauchiger_2018, Bong_2019}: under certain assumptions, one can prove that no joint probability distribution over all measurement outcomes can reproduce, as marginals, the accessible probabilities from the friend’s and Wigner’s perspectives. At the same time, we can encounter situations in which local perspectives are compatible with a global structure that is strictly more general. For example, in general relativity, local descriptions are approximately Minkowskian, yet the global structure can be a curved spacetime. In the present setting, perspective-dependent descriptions can indeed be obtained from the perspective-neutral formulation by appropriate reduction maps \cite{Vanrietvelde_2018a, Hoehn_2019_trinity, delaHamette_2021_perspectiveneutral}. Nevertheless, since the state relative to a (hypothetical) external reference frame carries more degrees of freedom than any single perspective accounts for, it is natural to ask the converse question: starting from the descriptions from single perspectives, can one reconstruct the description from the external, global perspective?\\

Here, we establish under which conditions, the global charge -- in the case of the translation group, the total momentum relative to an external reference frame -- can be inferred from internal perspectives alone. We do so by granting the internal observers progressively more resources, inspired by the different approaches to QRFs.   Our analysis shows concretely how the accessibility of these global properties depends fundamentally on the set of observables that are deemed accessible from within a given perspective, in line with previous arguments on the relation between the different approaches  \cite{Krumm_2021, Hoehn_2021, Carette_2023, Castro_Ruiz_2021, delaHamette_2021_perspectiveneutral,Hoehn_2023_subsystems, devuyst2025relation}, and whether one allows for classical communication. To achieve this, we first extend the perspectival and perspective-neutral formalisms to a new regime: that of a single, non-zero charge sector. We argue that, at least in the abelian case, there is no reason to single out the zero charge sector. After all, all one-dimensional irreducible representations of an abelian group differ only by a constant phase factor and choosing phase zero is no more fundamental than any other fixed phase. We therefore extend the perspectival and perspective-neutral frameworks for QRFs for the translation group to arbitrary fixed charge sectors. Interestingly, as we will see in the Discussion, this also turns out to be the most general form of a QRF transformation on the perspectival Hilbert spaces, compatible with transitivity and unitarity, ensuring invertibility and preservation of probabilities.\\

As for the structure of this paper, we will begin with the extension of both the perspectival and perspective-neutral frameworks (Section \ref{sec:extension}), deriving the QRF transformations, relative states, and invariant observables. We then proceed with the accessibility analysis in Section \ref{sec:accessibility}, studying three levels with progressively more resources, inspired by the operational and the perspectival/perspective-neutral approaches to QRFs. At the end of this section, we also briefly comment on the relation to the extra-particle approach. In the Discussion (Section \ref{sec:discussion}), we compare our results to recent literature in the field of QRFs and discuss potential generalisations of the extended QRF transformations. Finally, in Section \ref{sec:conclusion}, we conclude and give an outlook on open questions. 

\section{Extended Theoretical Framework}
\label{sec:extension}

Throughout this section, we will focus our attention on the one-dimensional translation group, $G=(\mathbb{R},+)$. Moreover, we will mainly consider three quantum systems $A$, $B$, and $C$ with Hilbert space $\H = \H_A\otimes\H_B\otimes \H_C \;\simeq\; L(\bbR)^{\otimes 3}$, each transforming under translations via the left regular representation. We will treat $A$ and $C$ as two ideal quantum reference frames, i.e.~they carry the regular representation of $G$, and $B$ as an additional, possibly composite, quantum system transforming under translations. These systems are prepared by an external referee Eve in a configuration of her choice. We will often refer to Eve, or $E$, as the external reference frame to which the internal systems do not have access. The quantity which is conserved under all symmetry transformations and which commutes with any other observable is the \emph{total centre of mass momentum} of $A$, $B$, and $C$, that is,  $\hat{P} = \hat{p}_A+\hat{p}_B+\hat{p}_C$.

\subsection{Perspectival Approach Extension} \label{sec:extension_perspectival}

The perspectival approach works directly at the level of the individual perspectives of the QRFs. The intuition is that, when describing other physical systems using standard quantum mechanics, we are implicitly doing so relative to a particular reference frame. Thus, strictly adhering to what is operationally accessible, the starting point is set directly within one quantum perspective, e.g.~the perspective of $A$ on $B$ and $C$. Interesting phenomena arise when one considers changing between different perspectives, such as the frame-dependence of coherence and entanglement \cite{Giacomini_2017_covariance, Vanrietvelde_2018a,delaHamette_2020,hoehn2021quantum, delaHamette_2021_perspectiveneutral, Hoehn_2023_subsystems, Kabel_2024, cepollaro2024}. To motivate the transformation between two different reference frames, $A$ and $C$, let us consider first what happens to the canonical variables \cite{Giacomini_2017_covariance}. For spatial reference frames, a change $\hat{S}_{A\to C}$ from $A$ to $C$ consists of a change of relative position coordinates of the form:
\begin{align}
\hat{x}_B &\mapsto \hat{S}_{A\to C}[\hat{x}_B]=\hat{q}_B - \hat{q}_A, \tag{1a}\label{eq:1a}\\
\hat{x}_C &\mapsto \hat{S}_{A\to C}[\hat{x}_C] = -\,\hat{q}_A \tag{1b}\label{eq:1b},
\end{align}
where we use $\hat{x}$ to denote positions relative to $A$ and $\hat{q}$ for positions relative to $C$. To obtain the transformation of the momentum operators, we have to appeal to the canonical commutation relations. Requiring that $[\hat{S}_{A\to C}[\hat{x}_i],\hat{S}_{A\to C}[\hat{p}_j]]=i\delta_{ij}$, with $i,j=B,C$ and $\hbar=1$, and assuming that the transformation is linear in the canonical variables, we find that (up to constants), the momenta transform as
\begin{align}
\hat{p}_B &\mapsto \hat{S}_{A\to C}[\hat{p}_B]= \hat{\pi}_B, \tag{1c}\label{eq:1c}\\
\hat{p}_C &\mapsto \hat{S}_{A\to C}[\hat{p}_C] = -\bigl(\hat{\pi}_A + \hat{\pi}_B\bigr), \tag{1d}\label{eq:1d}
\end{align}
where $\hat{p}$ and $\hat{\pi}$ refer to momenta in $A$'s and $C$'s frame, respectively. Note that these new canonical variables also give rise to a particular tensor product structure and thus a new notion of subsystem \cite{Zanardi_2001}, which is at the heart of the frame-dependence of many interesting quantities, such as entanglement \cite{Giacomini_2017_covariance, Vanrietvelde_2018a,delaHamette_2020,hoehn2021quantum, delaHamette_2021_perspectiveneutral, Hoehn_2023_subsystems, Kabel_2024, cepollaro2024}. As shown in \cite{Giacomini_2017_covariance}, the quantum operator implementing these canonical transformations takes the form of a quantum-controlled translation:
\setcounter{equation}{1}
\begin{align}
    \hat{S}_{A\to C} = \hat{\mathcal{P}}_{AC} e^{i\hat{x}_C \hat{p}_B}\label{eq: QRF0trafo} :\ \H_B^{(A)} \otimes \H_C^{(A)} \to \H_A^{(C)} \otimes \H_B^{(C)}
\end{align}
with $\hat{\mathcal{P}}_{AC}=\int dx\ _A\ket{-x}\bra{x}_C$ the parity-swap operator between $A$ and $C$ (cf.~\cite{Giacomini_2017_covariance}) and $\H_i^{(j)}$ denoting the state space of system $i$ relative to system $j$.
Note that Eq.~\eqref{eq:1d} states that, after transforming to $C$’s reference frame, the momentum assigned to system $C$ equals the negative of the total momentum of the \enquote{rest of the universe} (systems $A$ and $B$). This, in turn, means that the total momentum of all systems, as seen from an external reference frame, is zero -- precisely the starting point of the perspective-neutral approach (see Section \ref{sec:extensionPN}).\footnote{More precisely, we identify the momentum of a system relative to a QRF with its momentum as measured from an external reference frame. This identification is justified because our symmetry group includes only spatial translations, not boosts: while absolute positions are unobservable, momenta remain well-defined absolute quantities in each perspective (see also the discussion at the beginning of Section \ref{sec:accessibility}).} If, instead, the last line were replaced by
\begin{align}
\hat{p}_C \mapsto -\bigl(\hat{\pi}_A+\hat{\pi}_B\bigr)+P\tag{1d'}\label{eq:1d'},
\end{align}
for some constant $P\in \mathbb{R}$, the same reasoning would lead to the conclusion that the total momentum of the three systems is $P$ -- without affecting the canonical commutation relations, which remain unchanged by the addition of a constant. This modified canonical transformation is implemented by the \emph{generalised QRF transformation operator} for translations from $A$ to $C$ for three systems $A,B,$ and $C$ with total momentum $P$: 
\begin{align}
    \hat{S}^{P}_{A\to C} = \hat{\mathcal{P}}_{AC} e^{i\hat{x}_C (\hat{p}_B-P)}.\label{eq: genQRFtrafo}
\end{align}

The generalised QRF transformation differs from the QRF transformation for zero total momentum by a quantum-controlled phase $e^{-i\hat{x}_CP}$, which depends on both the total momentum $P$ and the position $\hat{x}_C$ of the new reference frame with respect to the old frame. When applied to a coherent superposition of position states of $C$ relative to $A$,  it induces an additional relative phase between the corresponding positions of $A$ relative to $C$ compared to the standard QRF transformation.\\

Just like the standard QRF transformation, the generalised transformation $\hat{S}^P_{A\to C}$ constitutes a valid frame change. In particular, it respects the canonical commutation relations and, as shown in Appendix \ref{app:genQRFtrafosPersp}, is unitary and transitive. Moreover, it can be straightforwardly extended to the case of $N$ systems momentum $P$, for which $\hat{p}_3 + \dots + \hat{p}_N = P - \hat{p}_1 -\hat{p}_2 $. Then, the QRF change operator from subsystems $1$ to $2$ reads
\begin{align}
    \hat{S}^{P}_{1\to 2} = \hat{\mathcal{P}}_{12} e^{i\hat{x}_2 (\sum_{i=3}^N \hat{p}_i-P)}.
\end{align}
\\
The generalised QRF transformation $\hat{S}^P_{A\to C}$ can also be understood within the standard framework of zero total momentum by viewing the subsystem $BC$ as moving relative to a fourth system $D$ in state $\ket{-P}_D$. Applying the standard QRF transformation $\hat{S}_{A\to C}=\hat{S}^0_{A\to C}$ to the resulting product state, we get
\begin{align}
    \hat{S}^0_{A\to C}\ket{\psi}^{(A)}_{BCD} &= \mathcal{P}_{AC}e^{i\hat{x}_C(\hat{p}_B+\hat{p}_D)}\ket{\psi^P}^{(A)}_{BC}\ket{-P}_D = \mathcal{P}_{AC}e^{i\hat{x}_C(\hat{p}_B-P)}\ket{\psi^P}^{(A)}_{BC}\ket{-P}_D\nonumber \\ &= \left(\hat{S}^P_{A\to C}\ket{\psi^P}^{(A)}_{BC}\right)\ket{-P}_D.
\end{align}
That is, the standard QRF transformation on the three-party system $BCD$ acts just like the extended QRF transformation on the two-party subsystem $BC$ from the point of view of $A$ while leaving the state of $D$ unchanged. In essence, the additional momentum-dependent phase of the modified QRF transformation can be seen to arise through a phase-kickback from an additional system, moving with opposite momentum relative to the subsystem under consideration \cite{Valente2025semester}. This highlights the consistency of the QRF transformations for different values of the total momentum.\\ 

Finally, we can understand the additional phase in the generalised QRF transformation as a change in the representation of the translation group. To see this, note  that the standard QRF transformation operator can be written as
\begin{equation}
\hat{S}_{A\to C}=\hat{\mathcal{P}}_{AC}\hat{T}_B(-\hat{x}_C),
\end{equation}
where $\hat{T}_B(x)=e^{-i\hat{p}_Bx}$ is the regular representation of the translation group on $\mathcal{H}_B^{(A)}$. Similarly, we can write
\begin{equation}
\hat{S}^P_{A\to C}=\hat{\mathcal{P}}_{AC}\hat{T}^P_B(-\hat{x}_C),
\end{equation}
where $\hat{T}_B^P(x)=e^{-i(\hat{p}_B-P)x}$.\footnote{Note that if there is more than one system in addition to the two QRFs, it is the representation on their \emph{joint} Hilbert space that changes. For example, if we have four systems $A,B,C,D$, we would get $\hat{T}_{BD}^P(x)=e^{-i(\hat{p}_B+\hat{p}_D-P)}$.} If we work on the projective Hilbert space, i.e.~identifying states that differ by a global phase, $\hat{T}_B^P(x)$ is actually equivalent to the regular translation representation $\hat{T}_B(x)$. On the Hilbert space $\mathcal{H}_B^{(A)}$, however, it constitutes an inequivalent representation.\footnote{Note further that, while the representation $\hat{T}^P(x) = e^{-ix(\hat{P}-P)}$ is still a strict representation satisfying $\hat{T}^P(x)\hat{T}^P(y) = \hat{T}^P(x+y)$ thanks to the abelian nature of the translation group, the situation is less straightforward when working with non-abelian groups. In fact, representations corresponding to different charge sectors in the non-abelian case are inequivalent even when working with projective Hilbert spaces.} The relation between the extended QRF transformation and different representations of the translation group will become even clearer when taking on a global or perspective-neutral point of view, since the total momentum characterises the different charge sectors of the representation of the translation group on the total Hilbert space $\mathcal{H}_{\kin}$.

\subsection{Perspective-Neutral Approach Extension}
\label{sec:extensionPN}

In this section, we show that the same extension arises naturally within the perspective-neutral approach to QRFs \cite{Vanrietvelde_2018a, Vanrietvelde_2018b, Höhn_2018a, Hoehn_2018b, Hoehn_2019_trinity, Hoehn_2020, Krumm_2021, hoehn2021quantum, Hoehn_2021, delaHamette_2021_perspectiveneutral, delaHamette_2021_entanglement, Hoehn_2023_subsystems, Chataignier_2024, Araujoregado_2025} through a modification of the constraint. In this approach, the perspectives relative to different choices of reference frame are obtained in two steps: Starting from a global point of view, one first implements invariance under the action of the symmetry group before conditioning on a particular state of the chosen QRF. Standardly, translation invariance of the total system is imposed by setting the total momentum to zero. This is analogous to the procedure in the Dirac quantisation of gauge theories, where one imposes the constraint which generates the gauge transformations on the quantum states \cite{Henneaux_Teitelboim_1992}.\\

To gain some intuition for this procedure, let us assume that Eve prepares systems $A$, $B$, and $C$ in the state
\begin{equation}
\ket{\psi_\kin}=\frac{1}{\sqrt{2}}\left(\ket{x_1}_A +\ket{x_2}_A\right)\ket{\phi}_B \ket{z}_C, \label{eq:kinematicalstate}
\end{equation}
where $\ket{\phi} = \int dy \phi(y) \ket{y}_B$ is an arbitrary state in $\mathcal{H}_B$ and $x,y$, and $z$ can be understood as the \enquote{absolute} positions of $A, B$, and $C$ relative to Eve. This state is referred to as the \emph{kinematical} state, and the corresponding Hilbert space is denoted by $\mathcal{H}_{\kin}$, to distinguish it from the \emph{physical} Hilbert space $\mathcal{H}_{\phys}$ obtained by imposing the constraint. We denote the momenta of the respective systems by $p_A$, $p_B$, and $p_C$. Imposing the constraint $\hat{P} = 0$ gives the physical state
\begin{align}
    \ket{\psi_\phys}&_{ABC}\equiv \delta(\hat{P})\ket{\psi_\kin} = \delta(\hat{p}_A+\hat{p}_B+\hat{p}_C)\frac{1}{\sqrt{2}}\left(\ket{x_1}_A +\ket{x_2}_A\right)\ket{\phi}_B \ket{z}_C \\ 
    &= \frac{1}{\sqrt{2}}\int dp_A dp_B dp_C \delta(p_A+p_B+p_C)\ \left(e^{-ip_A x_1}+e^{-ip_A x_2}\right) \int dy \phi(y)e^{-ip_B y}e^{-ip_C z} \ket{p_A}_A \ket{p_B}_B \ket{p_C}_C \\
    &= \frac{1}{\sqrt{2}}\int dp_A dp_B \left(e^{-ip_A (x_1-z)}+e^{-ip_A (x_2-z)}\right)\int dy \phi(y)e^{-ip_B(y-z)} \ket{p_A}_A \ket{p_B}_B\ket{-p_A-p_B}_C. 
\end{align}
From this \emph{perspective-neutral} state, we can obtain the respective states relative to the frames $A$ and $C$ by applying the Schrödinger reduction maps $\mathcal{R}_\S^{(A)}(X)=\bra{X}_A \otimes \mathbf{1}_{BC}$ and $\mathcal{R}_\S^{(C)}(Z)=\mathbf{1}_{AB} \otimes\bra{Z}_C$ \cite{delaHamette_2021_perspectiveneutral}. These can be understood as conditioning on frame $A$ being in position $X$ (relative to Eve) and frame $C$ being in  position $Z$ (relative to Eve), respectively. We thus obtain the \emph{relative physical states}
\begin{align}
     \ket{\psi}^{(A)}_{BC} &\equiv \mathcal{R}_\S^{(A)}(X) \ket{\psi_\phys}_{ABC} \nonumber\\ &= \frac{1}{\sqrt{2}}\int dy \phi(y)\left(\ket{y-(x_1-X)}_B\ket{z-(x_1-X)}_C\right. \left.+\ket{y-(x_2-X)}_B\ket{z-(x_2-X)}_C\right),\\
     \ket{\psi}^{(C)}_{AB} &\equiv \mathcal{R}_\S^{(C)}(Z) \ket{\psi_\phys}_{ABC} 
    =\frac{1}{\sqrt{2}}\left(\ket{x_1-(z-Z)}_A+\ket{x_2-(z-Z)}_A\right)\int dy \phi(y)\ket{y-(z-Z)}_B.
\end{align}

These quantum states are understood as those assigned, respectively, by $A$ to $B$ and $C$, and by $C$ to $A$ and $B$. In fact, for $X=Z=0$, they precisely correspond to the perspectival states relative to $A$ and $C$ in the previous section, related by $\hat{S}_{A\to C}$ \cite{Vanrietvelde_2018a}. In fact, the perspective-neutral and perspectival approaches have generally been shown to be equivalent for ideal QRFs \cite{delaHamette_2021_perspectiveneutral}. Importantly, the states $\ket{\psi}^{(A)}_{BC}$ and $\ket{\psi}^{(C)}_{AB}$ no longer carry any reference to the external structure $E$. The total momentum of systems $A$, $B$, and $C$, which is a global property defined with respect to Eve, is fixed to the particular value zero. While setting the constraint to zero is the common choice in the context of gauge theories (although this, too, may be questioned in the presence of boundaries, see, e.g.~\cite{Rovelli_2013_whygauge,Donnelly_2016}), we would like to stress that there is nothing that singles out a vanishing total momentum in the case of the translation group. In particular, if we understand the perspective-neutral state not just as an abstract mathematical construct but as the perspective of a (possibly fictitious) external observer, it is entirely natural to allow for the possibility of the total system moving at some fixed, non-zero momentum relative to her.\\

For these reasons, we want to allow Eve to prepare the systems with an arbitrary total momentum of her choice instead.  Formally, this corresponds to choosing a different global charge sector or, equivalently, modifying the constraint.\footnote{Note that the relation to the constraint in gauge theories is a little more subtle here since we are dealing with a physical symmetry rather than a gauge symmetry. The analogue of the translation invariance in gauge theory should thus be seen as the \emph{large} gauge transformations, which have a non-trivial action at the boundary of the spacetime region of interest. These are \emph{physical} symmetries and can thus be generated by a non-vanishing constraint.} Thus, consider again the example where Eve prepares the state \eqref{eq:kinematicalstate} but now fixes the total momentum of $A, B,$ and $C$ to the constant value $P$. The physical state is now obtained by imposing $\hat{P}=P$, yielding
\begin{align}
    &\ket{\psi_\phys^P}_{ABC}\equiv \delta(\hat{P}-P)\ket{\psi_\kin} = \delta(\hat{p}_A+\hat{p}_B+\hat{p}_C-P)\frac{1}{\sqrt{2}}\left(\ket{x_1}_A +\ket{x_2}_A\right)\ket{\phi}_B \ket{z}_C \\ 
    &= \frac{1}{\sqrt{2}}\int dp_A dp_B dp_C \delta(p_A+p_B+p_C-P) \left(e^{-ip_A x_1}+e^{-ip_A x_2}\right) \int dy \phi(y)e^{-ip_B y}e^{-ip_C z} \ket{p_A}_A \ket{p_B}_B \ket{p_C}_C \\
    &= e^{-iPz} \frac{1}{\sqrt{2}}\int dp_A dp_B \left(e^{-ip_A (x_1-z)}+e^{-ip_A (x_2-z)}\right)\int dy \phi(y)e^{-ip_B(y-z)} \ket{p_A}_A \ket{p_B}_B\ket{P-p_A-p_B}_C.\label{eq:psiPhysP} 
\end{align}
We observe that the physical state in a non-zero charge sector differs from that in the zero charge sector by a global phase $e^{-iPz}$. Moreover, the momentum of $C$ is now shifted by the total momentum, in line with the transformation properties of the canonical variables in Eq.~\eqref{eq:1d'}.\\

Before restricting to the perspectives of $A$ and $C$, let us return to the connection with representation theory, which we briefly alluded to at the end of the last section. Let us start from the kinematical Hilbert space $\mathcal{H}_{\kin}$. Note that we are generally dealing with the regular representation on the three kinematical systems, which decomposes into the direct sum over all charge sectors as
\begin{align}
    \mathcal{H}_{\kin}=\bigoplus_P \H_P.
\end{align}
In the usual perspective-neutral approach, the constraint $\hat{P}=0$ leads to a restriction to the zeroth charge sector $\mathcal{H}_0$ and is implemented by coherently averaging over all translations generated by $U^0(x)$ -- after all $\int dx\ U^0(x)\ket{\psi_{\kin}}=\int dx\ e^{i\hat{P}x}\ket{\psi_{\kin}}=\delta(\hat{P})\ket{\psi_{\kin}}$. More generally, if we want to restrict to charge sector $P$, we have to use $U^P(x) = e^{i(\hat{P}-P)x}$ instead.\footnote{The \enquote{extra-particle} approach goes one step further and considers the possibility of acting on all charge sectors with corresponding representations. We will go into more detail on this approach in Section \ref{sec: extraParticle}.} We thus apply the coherent $G$-twirl\footnote{The coherent $G$-twirl is sometimes referred to as the strong $G$-twirl or coherent group-averaging.} based on representation $U^P(x)$ to obtain the physical state
\begin{align}
   \ket{\psi_\phys^P}_{ABC}& =\mathcal{G}^P_{coh}[\ket{\psi_{\kin}}] \equiv\int dx U^P(x) \ket{\psi_{\kin}}\\ &= \int dx e^{ix(\hat{p}_A+\hat{p}_B + \hat{p}_C-P)}\frac{1}{\sqrt{2}}\left(\ket{x_1}_A +\ket{x_2}_A\right)\ket{\phi}_B \ket{z}_C \\
   &= \frac{1}{\sqrt{2}}\int dx e^{-ixP}\left(\ket{x_1-x}_A+\ket{x_2-x}_A\right)\int dy \phi(y) \ket{y-x}_B\ket{z-x}_C,
\end{align}
which can be seen to coincide with Eq.~\eqref{eq:psiPhysP}, thereby projecting from the kinematical Hilbert space onto the single charge sector $\mathcal{H}_P$. We can thus implement the constraint by coherently averaging over all possible global translations, just as in the usual perspective-neutral construction. The only difference is that the generator of translations deviates from the representation $U(x) = e^{ix\hat{P}}$ used in the standard formalism by a phase.\\

Let us now consider the two internal frame perspectives, $A$ and $C$, respectively. We apply the Schrödinger reduction maps $\mathcal{R}_\S^{(A)}(X)=\bra{X}_A \otimes \mathbf{1}_{BC}$ and $\mathcal{R}_\S^{(C)}(Z)=\mathbf{1}_{AB} \otimes\bra{Z}_C$ and obtain the respective relative physical states:
\begin{align}
     \ket{\psi^P}^{(A)}_{BC} \equiv \mathcal{R}_\S^{(A)}(X) \ket{\psi_\phys^P}_{ABC} = \frac{1}{\sqrt{2}}\int dy \phi(y)&\left(e^{-i(x_1-X)P}\ket{y-(x_1-X)}_B\ket{z-(x_1-X)}_C\right.\nonumber\\
    &\left.+e^{-i(x_2-X)P}\ket{y-(x_2-X)}_B\ket{z-(x_2-X)}_C\right),\label{eq:psiRelAX}\\
     \ket{\psi^P}^{(C)}_{AB} \equiv \mathcal{R}_\S^{(C)}(Z) \ket{\psi_\phys^P}_{ABC} 
    =e^{-i(z-Z)P}\frac{1}{\sqrt{2}}&\left(\ket{x_1-(z-Z)}_A+\ket{x_2-(z-Z)}_A\right)\int dy \phi(y)\ket{y-(z-Z)}_B.\label{eq:psiRelCZ}
\end{align}

These can be seen to be related by the \emph{generalised QRF transformation operator}
\begin{align}
   \hat{S}^P_{A\to C}(X,Z)=\int dz \ket{X+Z-z}_A\bra{z}_C\otimes \hat{U}_B^P(Z-z))=\int dz \ket{X+Z-z}_A\bra{z}_C\otimes e^{i(z-Z)(\hat{p}_B-P)}.\label{eq:genQRFtrafoPN}
\end{align}
In Appendix \ref{app:genQRFtrafosPN}, we show that this also holds for general relative physical states. Note that, commonly, one sets the reference frame to be at the origin relative to itself, i.e.~$X=Z=0$, such that
\begin{align}
    \ket{\psi^P}^{(A)}_{BC} &\equiv \mathcal{R}_\S^{(A)}(0) \ket{\psi_\phys^P}_{ABC} 
    =\frac{1}{\sqrt{2}}\int dy \phi(y)\left(e^{-ix_1P}\ket{y-x_1}_B\ket{z-x_1}_C+e^{-ix_2P}\ket{y-x_2}_B\ket{z-x_2}_C\right),\label{eq:psiRelAP}\\
    \ket{\psi^P}^{(C)}_{AB} &\equiv \mathcal{R}_\S^{(C)}(0) \ket{\psi_\phys^P}_{ABC} 
    =e^{-izP}\frac{1}{\sqrt{2}}\left(\ket{x_1-z}_A+\ket{x_2-z}_A\right)\int dy \phi(y)\ket{y-z}_B.\label{eq:psiRelCP}
\end{align}
These states can equivalently be understood as the perspectival states relative to $A$ and $C$ for non-zero total momentum $P$ (cf.~\cite{Vanrietvelde_2018a, delaHamette_2021_perspectiveneutral}). Consequently, the QRF transformation operator relating them must be equivalent to Eq.~\eqref{eq: genQRFtrafo} -- as seen straightforwardly by setting $X=Z=0$ in Eq.~\eqref{eq:genQRFtrafoPN}:
\begin{align}
    \hat{S}_{A\to C}^P(0,0) =\int dz \ket{X+Z-z}_A\bra{z}_C\otimes \hat{U}_B^P(Z-z))=\int dz \ket{-z}_A\bra{z}_C\otimes e^{iz(\hat{p}_B-P)}=\hat{\mathcal{P}}_{AC}e^{i\hat{x}_C(\hat{p}_B-P)}.
\end{align}
Note that the relative states in Eqs.~\eqref{eq:psiRelAP} and \eqref{eq:psiRelCP} again differ in their phases from those obtained in the zero charge sector. These phases depend on the positions $x_1$ and $x_2$ resp.~$z$ of the chosen reference frame relative to Eve. While this amounts to a global phase in the state relative to $C$, which we chose to be localised relative to the external frame, it is a {\em relative phase} for the reference frame $A$, which we considered to be in a superposition of locations $x_1$ and $x_2$ relative to $E$. Thus, for general QRFs, this phase becomes in principle observable. We discuss to what extent the phase can actually be observed in more detail in the next section.\\

\begin{figure}[h!]
    \centering
    \begin{tikzpicture}[font=\large]
         \node[above, fill=gray!15, anchor=south, rounded corners] at  (0,3){$\mathcal{H}_{\kin}=\mathcal{H}_A\otimes\mathcal{H}_B\otimes \mathcal{H}_C$};
        \node[above, fill=gray!15, anchor=south, rounded corners] at (0,1) {$\mathcal{H}_{\phys}$};
        \node[below, fill=gray!15, anchor=north, rounded corners] at (-2.5,0) {$\mathcal{H}^{(A)}_{BC}$};
        \node[below, fill=gray!15, anchor=north, rounded corners] at (2.5,0) {$\mathcal{H}^{(C)}_{AB}$};

        \draw[->, purple] (0,2.9) -- (0,1.8) node[right] at (0.2,2.3) {$\color{purple}{\Pi_{\phys}=\delta(\hat{P}-P)}$};
        \draw[->] (-0.2,0.8) -- (-1.8,0) node[left] at (-0.9,0.7) {$\mathcal{R}^{(A)}_\mathbf{S}$};
        \draw[->] (0.2,0.8) -- (1.8,0) node[right] at (0.9,0.7) {$\mathcal{R}^{(C)}_\mathbf{S}$};
    \end{tikzpicture}
    \caption{\emph{Extended perspective-neutral construction.} The construction follows the usual perspective-neutral approach by starting from the kinematical Hilbert space $\mathcal{H}_{\kin}$, projecting on the physical Hilbert space $\mathcal{H}_{\phys}$, and reducing into the perspective of $A$ and $C$, respectively, using the reduction maps $\mathcal{R}^{(A)}_\mathbf{S}$ and $\mathcal{R}^{(C)}_\mathbf{S}$. Our modification lies in the first step (highlighted in purple) through a different choice of projector $\Pi_{\phys}=\delta(\hat{P}-P)$ for non-zero total momentum $P$.}
    \label{fig:placeholder}
\end{figure}
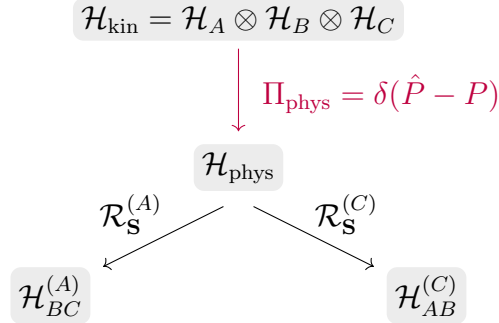

Let us also highlight again that this derivation of states in non-zero charge sectors is perfectly consistent with the framework for physical states in the zero charge sector, upon adding an additional system $D$, whose momentum exactly cancels the overall momentum $P$ of the subsystem $ABC$. We can thus see the non-zero charge sector as arising when some external agent, Eve, prepares the subsystem $ABC$ to move with momentum $P$ relative to some other system $D$, while the combined system remains in the $P=0$ charge sector. That is, she prepares the state
\begin{equation}
    \ket{\psi_{\kin}}^{(E)}_{ABCD} = \frac{1}{\sqrt{2}}(\ket{x_1}_A+\ket{x_2}_A)\ket{y}_B\ket{z}_C\ket{-P}_D.
\end{equation}
The physical state is obtained through the standard coherent $G$-twirl on the four-party system, i.e.
\begin{equation}
    \ket{\psi}^{(E)}_{ABCD} = \int dx \hat{U}_{0}(x)\ket{\psi_{\kin}}^{(E)}_{ABCD} = \frac{1}{\sqrt{2}}\int dx (\ket{x_1+x}_A+\ket{x_2+x}_A)\ket{y+x}_B\ket{z+x}_Ce^{ixP}\ket{-P}_D.
\end{equation}
Going into the perspective of $A$ (with $X=0$), we obtain
\begin{equation}
    \ket{\psi}^{(A)}_{BCD} = \frac{1}{\sqrt{2}}(e^{-ix_1P}\ket{y-x_1}_B\ket{z-x_1}_C+e^{-ix_2P}\ket{y-x_2}_B\ket{z-x_2}_C)\ket{-P}_D= \ket{\psi^P}^{(A)}_{BC}\ket{-P}_D.\label{eq:psiRelAD}
\end{equation}
Note that the state of $D$ still factorises out in the perspective of $A$. Tracing out the additional system, we thus simply obtain the physical state of $B$ and $C$ relative to $A$ in the non-zero charge sector (cf.~Eq.~\eqref{eq:psiRelAP}). A similar procedure for $C$ (with $Z=0$) yields
\begin{equation}
    \ket{\psi}^{(C)}_{ABD}=\frac{1}{\sqrt{2}} (\ket{x_1-z}_A+\ket{x_2-z}_A)\ket{y-z}_Be^{-izP}\ket{-P}_D = \ket{\psi^P}^{(C)}_{AB}\ket{-P}_D\label{eq:psiRelCD}
\end{equation}
and thus reproduces Eq.~\eqref{eq:psiRelCP} upon tracing out $D$. Note that since $D$ factorises out in both frames (thus avoiding the paradox of the third particle \cite{Krumm_2021}), it can be traced out without losing coherence (cf.~\cite{Valente2025semester}).

\subsection{Invariant Observables} \label{sec:invariantobservables}

A natural question that arises in the context of the generalised QRF transformation $\hat{S}^P_{A\to C}$ is whether observables that are invariant under the transformation for $P=0$ remain invariant if we consider non-zero total momentum. The answer is characterised by the following theorem.

\begin{theorem}
\label{th: invariance}
   Given a Hermitian operator with Schmidt decomposition $\hat{O}_{BC}^{(A)}=\sum_{k=1}^{n}\lambda_k \hat{O}^{(k)}_B\otimes \hat{O}^{(k)}_C\in \mathcal{L}(\mathcal{H}_{B}^{(A)}\otimes \mathcal{H}_{C}^{(A)})$, which is invariant under the QRF transformation from $A$ to $C$ for zero charge sector, i.e.~$\hat{S}_{A\to C} \hat{O}_{BC}^{(A)} \hat{S}^\dagger_{A\to C} =\mathcal{P}_{AC} \hat{O}_{BC}^{(A)}\mathcal{P}_{AC}^\dagger$, it is invariant under the generalised QRF transformation $\hat{S}^P_{A\to C}$ from $A$ to $C$ if and only if $[\hat{x}_C,\hat{O}_C^{(k)}]=0$ for all $k\in\{1,\dots,n\}$. 
\end{theorem}

The proof is provided in Appendix \ref{app: invariantOperators}. Note that the theorem extends straightforwardly to more systems with Hilbert space $\mathcal{H}^{(A)}_{BCD\dots}$ by subsuming all systems except for $C$ into one large subsystem $B'=BD\dots$.  Importantly, Theorem \ref{th: invariance} implies that not all operators which are invariant under the standard QRF transformation are also invariant under the generalised transformation for non-zero total momentum. For simple operators, which can be written as a tensor product of observables on $\mathcal{H}_B^{(A)}$ and $\mathcal{H}_C^{(A)}$, respectively, we can characterise the conditions under which an operator will fail to be invariant under the generalised QRF transformation by the following proposition, also proven in Appendix \ref{app: invariantOperators}.

\begin{proposition}
Given a non-zero Hermitian operator of the form $\hat{O}_{BC}^{(A)}=\hat{O}_B\otimes \hat{O}_C\in \mathcal{L}(\mathcal{H}_{B}^{(A)}\otimes \mathcal{H}_{C}^{(A)})$ that is invariant under the QRF transformation from $A$ to $C$ for zero charge sector, this invariance does \emph{not} extend to the transformation for non-zero charge sector if and only if $\{\hat{p}_B,\hat{O}_B\}=0$ and $\{\hat{x}_C,\hat{O}_C\}=0$.
\end{proposition}

A concrete example is provided by the parity operator.

\begin{example}
An example for an operator invariant under the standard but non-invariant under the generalised QRF transformation is the parity operator $\hat{\Pi}_B \otimes \hat{\Pi}_C$, with $\hat{\Pi} =\int dx |-x\rangle \langle x|$, which swaps the signs of both position and momentum. This satisfies $\{\hat{p}_B,\hat{\Pi}_B\}|\psi\rangle=(\hat{p}_B+\hat{\Pi}_B\hat{p}_B\hat{\Pi}_B^\dagger )\hat{\Pi}_B|\psi\rangle = (\hat{p}_B-\hat{p}_B)|\psi\rangle=0$ as well as $\{\hat{x}_C,\hat{\Pi}_C\}|\psi\rangle= (\hat{x}_C-\hat{x}_C)|\psi\rangle=0$. In particular, this implies invariance under the standard QRF transformation:
\begin{equation}
    e^{i\hat{x}_C\hat{p}_B}\hat{\Pi}_B\otimes\hat{\Pi}_Ce^{-i\hat{x}_C\hat{p}_B}=\hat{\Pi}_B\otimes\hat{\Pi}_Ce^{i(-\hat{x}_C)(-\hat{p}_B)}e^{-i\hat{x}_C\hat{p}_B} = \hat{\Pi}_B\otimes\hat{\Pi}_C.
\end{equation}
Yet, the operator $\hat{\Pi}_C$ does not commute with $\hat{x}_C$. Instead, we have $[\hat{x}_C,\hat{\Pi}_C]|\psi\rangle = (\hat{x}_C-\hat{\Pi}_C\hat{x}_C\hat{\Pi}_C^\dagger)\hat{\Pi}_C|\psi\rangle=2\hat{x}_C\hat{\Pi}_C|\psi\rangle$. Consequently, the state is not invariant under the generalised QRF transformation for non-zero momentum $P$:
\begin{equation}
    e^{i\hat{x}_C(\hat{p}_B-P)}\hat{\Pi}_B\otimes\hat{\Pi}_Ce^{-i\hat{x}_C(\hat{p}_B-P)}=\hat{\Pi}_B\otimes\hat{\Pi}_Ce^{i(-\hat{x}_C)(-\hat{p}_B-P)}e^{-i\hat{x}_C(\hat{p}_B-P)} = \hat{\Pi}_B\otimes\hat{\Pi}_Ce^{2i\hat{x}_CP}.
\end{equation}
\end{example}

Let us consider one more example, which is invariant under both the standard and the extended QRF transformations, and which will play an important role later in the discussion.
\begin{example}
    Consider a four-partite system and the following observable relative to $A$: $\hat{O}_{BCD}^{(A)}=\hat{x}_B\otimes \mathbb{I}_{CD}-\mathbb{I}_{BC}\otimes \hat{x}_D \equiv \hat{x}_B-\hat{x}_D$. Clearly, neither $\hat{x}_B$ nor $\hat{x}_D$ are invariant on their own. Nevertheless, their difference is invariant under both the standard QRF transformation and the generalised QRF transformation since
    \begin{align}
        e^{i\hat{x}_C(\hat{p}_B+\hat{p}_D-P)}\hat{O}_{BCD}^{(A)}e^{-i\hat{x}_C(\hat{p}_B+\hat{p}_D-P)}&=e^{i\hat{x}_C\hat{p}_B}\hat{x}_{B}e^{-i\hat{x}_C\hat{p}_B}\ -\ e^{i\hat{x}_C\hat{p}_D}\hat{x}_{D}e^{-i\hat{x}_C\hat{p}_D}\\&=(\hat{x}_B+\hat{x}_{C})-(\hat{x}_D+\hat{x}_{C})=\hat{x}_B-\hat{x}_D.
    \end{align}
\end{example}
The case of relative distances in a three-partite system is a bit more subtle, since we can only consider relative distances between system $B$ and one of the reference frames. In fact, the relative distance $\hat{x}_B-\hat{x}_C$ is not invariant in the above sense, neither under the standard nor under the generalised QRF transformation (from $A$ to $C$):
\begin{align}
    e^{i\hat{x}_C\hat{p}_B}(\hat{x}_B\otimes \mathbb{I}_C- \mathbb{I}_B\otimes\hat{x}_C)e^{-i\hat{x}_C\hat{p}_B}=(\hat{x}_B+\hat{x}_C)-\hat{x}_C=\hat{x}_B.
\end{align}
Similarly, 
    \begin{align}
        e^{i\hat{x}_C(\hat{p}_B-P)}(\hat{x}_B\otimes \mathbb{I}_C- \mathbb{I}_B\otimes\hat{x}_C)e^{-i\hat{x}_C(\hat{p}_B-P)}=(\hat{x}_B+\hat{x}_C)-\hat{x}_C=\hat{x}_B.
    \end{align}
However, if one takes into account that the position operator $\hat{S}^P_{A\to C}(\hat{x}_B-\hat{x}_C)(\hat{S}^P_{A\to C})^\dagger=\hat{q}_B$ in the frame of $C$ really describes the position of $B$ relative to the position of $C$, which just happens to be zero in its own frame, we recover again the right transformation properties. The same analysis holds, of course, for $\hat{q}_B-\hat{q}_A$ in the frame of $C$. This subtle case of relative distances between system $B$ and a reference frame will become relevant again in the discussion of the covariance of physical laws in Section \ref{sec:CovariancePhysicalLaws}.

\section{Accessibility Analysis: What Can Internal Observers Access?} \label{sec:accessibility}

Let us now address the question of what can be inferred from the internal perspectives about the respective configurations and the total momentum. To this end, we introduce a game run by an external referee, Eve, and played by two observers, Alice and Charlie, each equipped with a quantum reference frame $A$ and $C$, respectively, as defined in Section~\ref{sec:extension}. In addition, Alice and Charlie each have a classical register to store their measurement outcomes.\footnote{We assume that the way Alice and Charlie store the measurement outcomes is not affected by the QRF transformation. This is ensured if we assume that the register is composed of \emph{internal} degrees of freedom that do not change under classical and quantum translations.} The game proceeds sequentially over multiple rounds: at the beginning of each, Eve prepares the QRFs $A$ and $C$ together with the quantum system $B$ in a single pure state and then projects it into a single charge sector of her choice, and repeats this same preparation across rounds. In each round, either Alice performs measurements on $B$ and $C$ or Charlie performs measurements on $A$ and $B$. The task for the internal observers is to determine  which total momentum sector they are prepared in and, possibly, other global properties about the state prepared by Eve.

\begin{figure}[h!]
    \centering
    \includegraphics[width=0.9\linewidth]{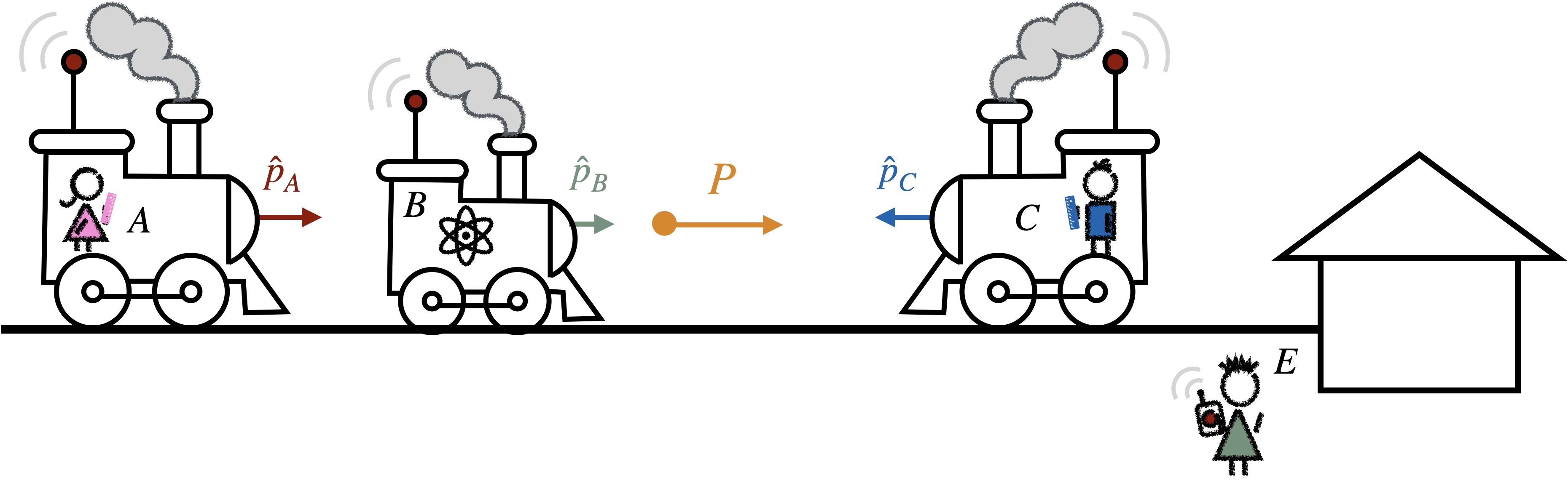}
    \caption{\emph{Illustration of the game run by an external referee Eve. At the beginning of each run, she prepares the reference frames $A$ and $C$, and an additional system $B$ in a configuration, imposing total momentum $\hat{P}=\hat{p}_A+\hat{p}_B+\hat{p}_C= P$, depicted as the (orange) arrow from the centre of mass. Then, either Alice or Charlie perform a measurement on the remaining systems on the tracks. Over time, they collect data, based on which they try to infer the total momentum $P$. Whether this strategy succeeds depends on the specific observables and resources available to them, characterised by three levels of increased accessibility in the main text.}
    \label{fig:sketch_game} }
\end{figure}

In what follows, we analyse what Alice and Charlie can infer 
by playing this game, gradually increasing the access and resources available to them. Before going into the technical details, let us anticipate a common source of confusion that stems from the rarity of a purely translation-invariant system. We can only speak about accessing the total momentum because, unlike many physical systems, the setup considered here does not involve an invariance under momentum translations -- we are dealing with the symmetry group $(\mathbb{R},+)$, not the (much more common) Galilei group. Accessing the total momentum thus does not break the relational spirit, which lies at the heart of the study of QRFs.

A physical example from the classical world to keep in mind in the subsequent discussion would be as follows: Consider a set of three trains, all moving at potentially different speeds on the tracks.\footnote{We thank Andrea Di Biagio for conceiving the original version of the train-track scenario illustrating systems that measure relative position and absolute momentum.} Each train represents one of the systems $A, B,$ and $C$. From within each train, one can send out light signals and thus measure the relative location of the other systems relative to oneself. Moreover, Alice and Charlie can each look at the other trains and infer their absolute momenta (with respect to the tracks), e.g.~by counting how many times the train wheels rotate in a given time frame. We assume that they cannot look down onto the wheels of their own train and thus cannot access their own absolute momentum. The tracks are taken to be perfectly homogeneous such that there is no visible position reference outside of the tracks relative to which the absolute location of any of the vehicles could be defined. Thus, we are dealing with a translation-invariant setup, in which there is nevertheless a sense of absolute momentum.\\

In what follows, let us assume that Eve prepares a state of the form
\begin{equation}\label{eq:psiKinEnt}
    \ket{\psi_{\kin}} = \frac{1}{\sqrt{2}}\left(\ket{x_1}_A\ket{z_1}_C+e^{i\varphi}\ket{x_2}_A\ket{z_2}_C\right)\ket{\phi}_B.
\end{equation}
The example state from the previous section corresponds to the special case $z_1=z_2 =z$. For greater generality, we now consider a scenario where Eve prepares $A$ and $C$ in an entangled state with a relative phase $\varphi$. In this case, after projecting on a single charge sector $P$ of her choice, the states in the perspectives of Alice and Charlie (conditioned on being at the origin) are 
\begin{align}
    \ket{\psi^P}^{(A)}_{BC}&=\frac{1}{\sqrt{2}} \int dy\phi(y)\left(e^{-ix_1P}\ket{y-x_1}_B\ket{z_1-x_1}_C+e^{i\varphi}e^{-ix_2P}\ket{y-x_2}_B\ket{z_2-x_2}_C\right)\nonumber\\
    &\equiv e^{-ix_1P}\frac{1}{\sqrt{2}}\left(\ket{\phi^{(1)}}^{(A)}_{BC}+e^{i\Phi^{(A)}}\ket{\phi^{(2)}}^{(A)}_{BC}\right),\label{eq:staterelAlice}\\
    \ket{\psi^P}^{(C)}_{AB}&=\frac{1}{\sqrt{2}} \int dy\phi(y)\left(e^{-iz_1P}\ket{x_1-z_1}_A\ket{y-z_1}_B+e^{i\varphi}e^{-iz_2P}\ket{x_2-z_2}_A\ket{y-z_2}_B\right)\nonumber\\
    &\equiv e^{-iz_1P}\frac{1}{\sqrt{2}}\left(\ket{\tilde{\phi}^{(1)}}^{(C)}_{AB}+e^{i\Phi^{(C)}}\ket{\tilde{\phi}^{(2)}}^{(C)}_{AB}\right),\label{eq:staterelCharlie}
\end{align}
where
\begin{align}
    \Phi^{(A)}&\equiv \varphi + (x_1-x_2)P \label{eq: phase relative to A},\\
    \Phi^{(C)}&\equiv \varphi + (z_1-z_2)P \label{eq: phase relative to C}.
\end{align}

We will use this example to discuss to what extent the internal observers can infer the values of the total momentum $P$ and the relative phase $\varphi$ \enquote{from inside}. We do so by progressively granting them access to more resources.

\subsection{Level 1: Relative Framed Observables Only}
\label{sec: level1}

At the most restrictive level, consider the case in which the internal observers Alice and Charlie have access only to \emph{relative framed} observables as defined in the operational approach to QRFs \cite{Loveridge_2017, Carette_2023}. Adapting the definitions to our three-particle setup, the \enquote{$A$-relative $C$-framed observables} of $B$ and $C$ relative to $A$ are observables of the form\footnote{To be precise, in \cite{Carette_2023}, the observables are defined as bounded operators of the form $\mathcal{O}_{BC,\op}^{(A)}=\yen^A\left(O_B\otimes E_C(X)\right)$, where the frame observable $E_C:B(G)\to\mathcal{B}(\mathcal{H}_C)$ is a \emph{covariant} POVM satisfying $E(g.X)=U(g)XU(g)^\ast$ for all $g\in G$ and $X\in B(\Sigma)$. The $\yen$-map, in turn, relativises this operator. Of course, the position operator is not a bounded operator on $L^2(\mathbb{R})$, hence the adapted definition above. Moreover, in the present work, we restrict ourselves to three copies of $L^2(G)$, corresponding to one system and two \emph{ideal} frames -- that is, two frames which are principal and sharp (cf.~\cite{Carette_2023}). This implies that these frames carry the covariant PVM of $G$, 
$E: B(\Sigma) \to \mathcal{L}(\mathcal{H}), \ 
E(X) = \int_X \ket{g}\bra{g}\, \mathrm{d}\mu(g), \ X \in B(\Sigma),$
with $\Sigma \cong G$ as manifolds, and $\mu$ the counting measure for discrete $G$ or the Haar measure for continuous $G$.}
\begin{align}
   \mathcal{O}_{BC,\op}^{(A)}=\int dx \ket{x}\bra{x}_A \otimes U^P_{BC}(x) \left(O_B\otimes |y\rangle\langle y|_C\right) U^P_{BC}(x)^\dagger
\end{align}
and linear combinations thereof, where $O_B\in \mathcal{L}(\mathcal{H}_B)$ and the projectors cover all $y\in \mathbb{R}$. These relative framed observables are elements of $\mathcal{L}(\mathcal{H}_A\otimes \mathcal{H}_B\otimes\mathcal{H}_C)^G$, the set of \emph{invariant} operators on the total Hilbert space.\footnote{Note that, strictly speaking, the $|x\rangle$ are not elements of $L^2(\mathbb{R})$ and thus the projector $|x\rangle\langle x|$ should be understood formally, to be made precise using methods from rigged Hilbert spaces \cite{gel2016generalized, schwartz1950theorie} (see also \cite{thiemannModernCanonicalQuantum2008, Giulini:1998kf, Giulini:1998rk} for methods on refined algebraic quantisation).} Moreover, the specific observables of the form $O_B \otimes |y\rangle\langle y|_C$ on $B$ and $C$ are called \emph{framed} observables. Similarly, relative framed observables relative to $C$  are of the form 
\begin{equation}
\mathcal{O}_{AB,\op}^{(C)} =\int dx\  U^P_{AB}(x)\left(|y\rangle\langle y|_A\otimes O_B\right) U^P_{AB}(x)^\dagger\otimes |x\rangle \langle x|_C .
\end{equation}
Importantly, note that $B$ is treated differently from $A$ and $C$: while the observers Alice and Charlie can measure any observable on $B$, they are more restricted when it comes to measurements on the other reference frame. For example, Alice can only measure $C$'s position but has no access to its momentum, and vice versa. In the absence of an additional system $B$, this means that only measurements of the relative orientation observable between the two frames are allowed \cite[pp.~30-31]{Carette_2023}. This assumption fits with the intuition behind the operational approach, in which the choice of QRF is \enquote{akin to fixing a measuring apparatus and choosing a particular pointer observable} \cite[p.~3]{Carette_2023}. By contrast, other QRF approaches, in particular the perspectival, the perspective-neutral, and the extra-particle approaches, allow more freedom in the choice of observable that can be measured on the other reference frame -- the state of the latter can, in general, be measured in multiple bases.

In the present scenario, the position of the reference frames $A$ resp.~$C$ plays a preferred role as the pointer observable relative to which the observables of $B$ are measured. Crucially, the individual momentum operators $\hat{p}_A$ and $\hat{p}_C$ of the reference frames themselves are \emph{not} relative framed observables, even though they commute with the total momentum. This explains why, in the operational approach, one cannot distinguish between a superposition or a mixture of position eigenstates of one reference frame relative to another \cite{Carette_2023} -- to verify such a superposition, one would need access to an incompatible observable, such as the momentum of the reference frame. The same reasoning applies to the question of whether Alice and Charlie can infer the total charge $P$. The latter is encoded in the relative phase $\Phi^{(A)}$ resp.~$\Phi^{(C)}$, which cannot be accessed by means of relative distance measurements as seen from reference frames $A$ and $C$ alone. \emph{Thus, given access to only relative framed observables, Alice and Charlie cannot infer the total charge.}

\subsection{Level 2: Relational Observables}
\label{sec: level2}

Let us next restrict the observables that the internal observers Alice and Charlie can measure to those that are strictly in the algebra of relational observables within the (extended) perspective-neutral approach. The observables that Alice can access at this level are defined relative to the \emph{kinematical} Hilbert space $\mathcal{H}_{\kin}=\mathcal{H}_A\otimes\mathcal{H}_B\otimes\mathcal{H}_C$ and are of the form
\begin{align}
   \mathcal{O}_{BC,\PN}^{(A)}=\int dx \ket{x}\bra{x}_A \otimes U^P_{BC}(x) (O_{BC}) U^P_{BC}(x)^\dagger
\end{align}
and linear combinations thereof, for $O_{BC}\in \mathcal{L}(\mathcal{H}_B\otimes 
\mathcal{H}_C)$ arbitrary observables on $B$ and $C$. Similarly, the observables that Charlie can access at this level are of the form
\begin{align}
    \mathcal{O}_{AB,\PN}^{(C)}=\int dx  U^P_{AB}(x) (O_{AB}) U^P_{AB}(x)^\dagger\otimes \ket{x}\bra{x}_C 
\end{align} and linear combinations thereof, for $O_{AB}\in \mathcal{L}(\mathcal{H}_A\otimes 
\mathcal{H}_B)$. Note here the crucial difference to Level 1: now, all linear operators $O_{BC}$ resp. $O_{AB}$ are allowed and not just the framed observables of the form $O_B\otimes |y\rangle\langle y|_C$ resp. $|y\rangle\langle y|_A\otimes O_B$. \\

Let us get a better understanding of the observables an internal QRF can thus access on the other subsystems. For this, we start with the position operator of $C$ relative to $A$, following Example 5 in \cite{delaHamette_2021_perspectiveneutral}. Denoting by $\hat{q}_A$ the position operator of $A$, we obtain the following relational observable encoding the position $C$, conditioned on the frame $A$ being at position zero:
\begin{align}
   F_{\hat{q}_C,A}(0)&=\int dx \ket{x}\bra{x}_A \otimes \mathbb{I}_B\otimes e^{-ix(\hat{p}_C-P)}\hat{q}_C e^{ix(\hat{p}_C-P)} \\
   &= \int dx \ket{x}\bra{x}_A \otimes \mathbb{I}_B \otimes (q_C+x) = \hat{q}_C-\hat{q}_A,
\end{align}
where, in the last line, we omitted the identities on the other subsystems for notational clarity. Next, let us check whether Alice can access the momentum operator of $C$:
\begin{align}
   F_{\hat{p}_C,A}(0)&=\int dx \ket{x}\bra{x}_A \otimes \mathbb{I}_B\otimes e^{-ix(\hat{p}_C-P)}\hat{p}_C e^{ix(\hat{p}_C-P)} \\
   &= \int dx \ket{x}\bra{x}_A \otimes \mathbb{I}_B\otimes \hat{p}_C = \hat{p}_C.
\end{align}
We find that the momentum of $C$ lies in the algebra of invariant observables, accessible to Alice. The same argument holds for the momentum of $A$ -- it is accessible to Charlie. Moreover, it is easy to see that both Alice and Charlie can access the momentum of $B$, $\hat{p}_B$. Yet, without communication, they cannot combine these results to infer the total momentum $P$.

More generally, we can consider rotated quadratures of the form $\cos(\theta)(\hat{q}_C - \hat{q}_A) + \sin(\theta)\hat{p}_C$ for arbitrary $\theta$ \cite{Vogel_1989}. In order to show that these are invariant, we follow the same strategy as in the perspective-neutral approach \cite{delaHamette_2021_perspectiveneutral}. That is, we use the fact that the coherent and incoherent $G$-twirls are equivalent on observables on the physical Hilbert space\footnote{The proof works for any unitary representation and thus in particular for $U^P(x)$.} \cite{delaHamette_2021_perspectiveneutral}, and show that under the incoherent modified $G$-twirl, the above observables transform as
\begin{align}
   \mathcal{G}_{incoh}^P[\cos(\theta)(\hat{q}_C - \hat{q}_A) + \sin(\theta)\hat{p}_C] &= \int dx U^P_{AC}(x) [\cos(\theta)(\hat{q}_C - \hat{q}_A ) + \sin(\theta)\hat{p}_C] U^P_{AC}(x)^\dagger \\
   &= \cos(\theta)(\hat{q}_C - \hat{q}_A) + \sin(\theta)\hat{p}_C,
\end{align}
showing that all such rotated quadratures are invariant under the $G$-twirl operation and therefore accessible to the internal observers. This access to all rotated quadratures enables Alice to perform {\em complete tomographic reconstruction} of Charlie's quantum state relative to Alice's reference frame. The key insight is that Alice can measure the full set of observables $\{\cos(\theta)(\hat{q}_C - \hat{q}_A) + \sin(\theta)\hat{p}_C : \theta \in [0, 2\pi)\}$, which constitutes a tomographically complete set for reconstructing Charlie's relative state. Evidently, the same argument holds for observables on system $B$ and for correlations between quadrature observables measured on $C$ and $B$. Therefore, Alice can perform complete tomography on the relative state $\ket{\psi}_{BC}^{(A)}$. Analogous reasoning implies that Charlie can perform complete tomography on $\ket{\psi}_{AB}^{(C)}$.\\

In the perspectival approach, the relational observables are encoded very naturally in the Hilbert space relative to the QRF: For ideal QRFs, which we are dealing with here, the set of relational observables as in the perspective-neutral approach is isomorphic to the set of {\em all} possible observables on the Hilbert space of $B$ and $C$ relative to $A$. That is, Alice has access to the observables
\begin{align}
\mathcal{O}_{BC,\persp}^{(A)}=O_{BC} \in \mathcal{L}(\mathcal{H}_B^{(A)}\otimes \mathcal{H}_C^{(A)}).
\end{align}
Similarly, Charlie has access to all observables 
\begin{align}
\mathcal{O}_{AB,\persp}^{(C)}=O_{AB} \in \mathcal{L}(\mathcal{H}_A^{(C)}\otimes \mathcal{H}_B^{(C)}).
\end{align}
This capability has always been implicit in the perspectival approach to quantum reference frames. When we write relative states as pure states in the literature, this implicitly assumes that the internal observers have access to observables that allow them to verify these pure state assignments. Without access to observables conjugate to position, such as momentum, one would be forced to assign mixed states to other systems, contradicting the standard treatment for perspectival QRFs. Note, though, that this is precisely what happens in the operational approach -- without access to the conjugate momentum, mixtures and pure states cannot be distinguished by the accessible observables \cite{Carette_2023, Castro-Ruiz_2025}.\\

So, Alice can perform complete tomography of the relative state of Charlie and system $B$ in both the perspectival and the perspective-neutral approach. While this lets her determine the phase $\Phi^{(A)}$ (see Eq.~\eqref{eq: phase relative to A}), she cannot recover the individual parameters entering that phase: the global charge sector $P$, the difference between her own absolute positions relative to Eve, $x_1-x_2$, and the relative phase $\varphi$. Since Alice cannot access Eve's frame, she cannot isolate the total momentum from her measurement results alone. In particular, her measurements depend statistically on $P$ but remain coupled to inaccessible absolute information, preventing direct inference of the charge sector value by individual observers. \emph{Thus, given access to all relational observables in the perspective-neutral approach viz.~all observables in the perspectival approach, but without communication between the reference frames, Alice and Charlie still cannot infer the total charge.}

\subsection{Level 3: Inter-Frame Communication}
\label{sec: level3}

At this level, we allow internal observers Alice and Charlie to communicate and share their measurement results from Level 2. This collaboration enables them to extract global information that was inaccessible to them individually.\\

In fact, now that Alice and Charlie can communicate classically, determining the value of the total momentum $P$ is rather straightforward. Since the global state is restricted to a fixed total-momentum sector, the physical-state constraint implies
$\langle \hat p_A \rangle + \langle \hat p_B \rangle + \langle \hat p_C \rangle = P$. Provided that Alice and Charlie are able to exchange their measurement results, Alice can estimate $\langle \hat p_C \rangle$ (and $\langle \hat p_B \rangle$) from her perspective, Charlie can estimate $\langle \hat p_A \rangle$ (and $\langle \hat p_B \rangle$) from his, and by communicating their results they can directly reconstruct $P=\langle P \rangle = \langle p_A \rangle + \langle p_B \rangle + \langle p_C \rangle$. The interest of the present analysis, however, does not lie in the mere determination of $P$, but in understanding which other global properties, beyond such linear constraints, can be operationally accessed from within the system. In particular, while not strictly necessary for determining $P$, extending the tomography-based strategy from Level 2 will allow us to understand better how information about global structure is encoded in the \emph{relation} between QRF perspectives: In the quantum case, this relation is captured by the $P$-dependent QRF transformation $\hat S^{(P)}_{A\to C}$, and the central question is whether internal observers can operationally reconstruct this full unitary relation between their perspectives.\\

From their respective relational tomographical capabilities at Level 2, Alice and Charlie can each determine phase information from the relative states they assign to each other and system $B$. To repeat, Alice measures a phase
\begin{align}
   \Phi^{(A)} &= \varphi +(x_1-x_2)P
\end{align}
while Charlie independently measures
\begin{align}
   \Phi^{(C)} &= \varphi +(z_1-z_2)P.
\end{align}
It is now relatively straightforward to devise a formula through which Alice and Charlie can combine their information to infer the total charge. In particular,  they can now compute the difference: 
\begin{align}
   \Delta \Phi &= \Phi^{(A)} - \Phi^{(C)} = P ((x_2-z_2)-(x_1-z_1)).
\end{align}
This phase difference $\Delta \Phi$ can be inferred by comparing their measurement results for $\Phi^{(A)}$ and $\Phi^{(C)}$ through classical communication. Crucially, the global phase $\varphi$ cancels out in this difference. Since Alice and Charlie can also measure the relative distances $(x_1-z_1)$ and $(x_2-z_2)$ through the corresponding relational observables $F_{\hat{q}_A,C}(0)$ or $-F_{\hat{q}_C,A}(0)$, they have access to both $\Delta \Phi$ and $((x_2-z_2)-(x_1-z_1))$. This allows them to extract the charge sector value:
\begin{align}\label{eq:PfromDeltaPhi}
   P = \frac{\Delta \Phi}{(x_2-z_2)-(x_1-z_1)},
\end{align}
where we assume that $x_1-x_2 \neq z_1-z_2$.\\

Although the phase-difference method cancels $\varphi$ and singles out $P$, it still presumes a particular preparation relative to Eve (a state of the form given in Eq.~\eqref{eq:psiKinEnt}). To see how Alice and Charlie can infer the total momentum $P$ even if they are prepared in an arbitrary state relative to Eve, consider the following general relative state of $B$ and $C$ with respect to $A$ (cf.~Eq.~\eqref{eq:GenRelStateA}) and rewrite it using $u=y-x$, $v=z-x$:
\begin{equation}
\ket{\psi}_{BC}^{(A)} = \int dx dy dz\ e^{-iPx}\psi(x,y,z) \ket{y-x}_B \ket{z-x}_C = \int du dv\ \Phi_A(u,v;P)
\ket{u}_B\otimes\ket{v}_C ,
\end{equation}
with $\Phi_A(u,v;P) = \int dx e^{-iPx} \psi(x,x+u,x+v)$. Similarly, the relative state of $A$ and $B$ with respect to $C$ (cf.~Eq.~\eqref{eq:genRelStateC}) can be rewritten using $s=x-z$, $t=y-z$ to obtain
\begin{equation}
\ket{\psi}_{AB}^{(C)} = \int dx dy dz\ e^{-iPz}\psi(x,y,z) \ket{x-z}_A \ket{y-z}_B = \int ds dt\ \Phi_C(s,t;P)
\ket{s}_A\otimes\ket{t}_B ,
\end{equation}
with $\Phi_C(s,t;P) = \int dz e^{-iPz}\psi(z+s, z+t, z)$. As discussed above, Alice and Charlie can each reconstruct these states via complete tomography. They only determine them up to global phases $e^{i\theta_A}$ and ~$e^{i\theta_C}$, respectively. A change of variables shows that
\begin{align} \label{eq:relationrelativestates}
e^{i\theta_C} \Phi_C(s,t;P) = e^{i(\theta_A+Ps)}\Phi_A(t-s,-s;P).
\end{align}

To determine $P$, Charlie can pick two distinct pairs $(s_i,t_i)$ (for $i=1,2$) for which $\Phi_C(s_i,t_i;P)$ and $\Phi_A(t_i-s_i,-s_i;P)$ do not vanish, communicate them to Alice, and together verify
\begin{align}
    R(s_i,t_i):= \frac{\Phi_C(s_i,t_i;P)}{\Phi_A(t_i-s_i,-s_i;P)} = e^{iPs_i}e^{i(\theta_A-\theta_C)}.
\end{align}
From the ratio of two such relations one obtains
\begin{align}
    P=\frac{i}{s_2-s_1}\log(\frac{R(s_1,t_1)}{R(s_2,t_2)})= \frac{\operatorname{Arg}R(s_{2},t_{2}) - \operatorname{Arg}R(s_{1},t_{1})}{s_{2}-s_{1}}.
\end{align}
If $\Phi_A$ and $\Phi_C$ are continuous and differentiable functions, it is convenient to define 
\begin{align}
   R(s,t) := \frac{\Phi_C(s,t;P)}{\Phi_A(t-s,-s;P)} = e^{iPs}e^{i(\theta_A-\theta_C)}. 
\end{align}
for those $(s,t)$ where $\Phi_A$ does not vanish. Assuming differentiability of $R$, we find
\begin{align}
    P = \frac{\partial}{\partial s}\arg R(s,t).
\end{align}

We thus conclude: whenever there exist at least two pairs $(s,t)$ with $s_1 \neq s_2$ for which both $\Phi_A$ and $\Phi_C$ are non-zero, Alice and Charlie can reconstruct $P$. The only case in which this fails is when $A$ and $C$ are perfectly localised relative to one another. In this case, the phase $e^{iPs}$ is unobservable. This highlights a specifically quantum feature of reference frames: unlike classical frames, which are always sharply localised, QRFs and the transformations between them can carry information about global properties such as $P$. This information can then be inferred locally and combined across different perspectives.\\

It is worth noting that the general relation of Eq.~\eqref{eq:relationrelativestates} can be obtained in two conceptually distinct ways: either from the global picture, using the extended perspective-neutral framework or by modifying the canonical transformations in the perspective-dependent formalism. The strategy presented here works because correlations between different internal perspectives isolate the global charge information while eliminating frame-dependent contributions. This collaborative access contrasts with the limitations of Levels 0 and 1, showing that communication between internal frames can reveal global properties hidden from individual observers.\\

Note, however, that while Alice and Charlie can identify the global charge sector, they cannot determine the relative phase $\varphi$ that characterises the superposition of $A$ and $C$ in Eve's preparation. This phase cancels out in their collaborative measurement strategy and remains fundamentally inaccessible from within the system. \emph{Thus, through collaboration via classical communication, internal observers can achieve complete charge sector identification but lack access to the phase $\varphi$.}\\

Note further that the accessibility game studied in this subsection presumes that Eve prepares the remaining systems in a fixed state. But what if they are not given the promise that Eve prepares the same state over and over again? Instead, she could prepare a collection of \emph{distinct} states in the same charge sector over multiple rounds. In this case, they could still infer the total momentum via process tomography of an unknown unitary on a continuous-variable system. Since the latter requires, in principle, an unbounded set of input-output probes (with precision improving as this set grows), Eve would have to prepare, for any target accuracy, a sufficiently large collection of linearly independent states that share the same total momentum across runs. Alice and Charlie then carry out full state tomography of each input-output state pair, reconstruct the unitary, and read off the selected charge. 

In Appendix \ref{app:perspectivalGame}, we introduce a purely perspectival game in which Alice sends states to Charlie, and they must guess which global charge their states are compatible with -- without ever assuming that the entire state is defined by a hypothetical Eve.\\

A final caveat remains: The above reasoning presumes pure states relative to $A$ and $C$, which can only come about if Eve prepares all the systems in a pure state. If, for instance, Eve were to prepare a mixed state of $A$, $B$, and $C$ in charge sector $P$, or if, for some reason, Alice and Charlie do not have access to all other subsystems, they would reconstruct mixed relative states and the charge $P$ cannot -- or can only partially -- be reconstructed.

\subsection{Relation to the 'Extra Particle' approach}
\label{sec: extraParticle}

There is yet another, more permissive, approach to QRFs: the extra-particle approach of \cite{Castro_Ruiz_2021, Castro-Ruiz_2025}. In this approach, one allows access to all invariant operators, including the global charge. That is, Alice or Charlie alone can access the charge sector in which they reside, by means of an \enquote{extra particle}. This additional degree of freedom ensures the invertibility of a more general set of QRF transformations than the ones obtained in the perspective-neutral or perspectival approach. In the case of the translation group, the extra particle is precisely the total momentum.\footnote{More generally, the extra particle is described by elements in the commutant of the algebra relative to the chosen reference frame within the algebra of invariant observables. One can understand it to capture the invariant properties of the reference frame itself \cite{Castro_Ruiz_2021}.} It is, by construction, in principle accessible to both Alice and Charlie so that they can infer the total momentum without communication and by, essentially, performing a single measurement: the state of the extra particle.
Applied to our setup, the Hilbert space relative to Alice in the extra particle approach is given by
\begin{equation}
\mathcal{H}^{|A}=\int^\oplus dp\ \mathcal{H}_B^{(p)}\otimes \mathcal{H}_C^{(p)}.
\end{equation}
Formally, one could also write the Hilbert space decomposition as $\mathcal{H}^{|A}=\int^\oplus dp\ \mathcal{H}_p\otimes \mathcal{H}_B^{(p)}\otimes \mathcal{H}_C^{(p)}$. However, the overall translation invariance implies a superselection rule that prevents (the verification of) superpositions of different total momenta $p$. Thus, the Hilbert space of the extra particle is simply $\mathcal{H}_p\simeq \mathbb{C}$, that is, merely a number that labels the charge sector rather than a proper quantum degree of freedom. A QRF transformation in the extra-particle approach can be seen to act separately on each charge sector via (\cite[p.~23]{Castro_Ruiz_2021}, adapted to our notation)\footnote{Note that $\hat{p}$ acts on $\H_p$ and should not be confused with the total momentum operator $\hat{P}$ acting on $\H_A\otimes \H_B \otimes \H_C$ -- since the QRF transformation $\hat{S}_{A\to C}$ maps between $\mathcal{H}^{|A}$ and $\mathcal{H}^{|C}$, this would not even be well-defined. Moreover, if one tried to define the map on the kinematical Hilbert space, it would not be self-adjoint.}
\begin{equation}
\hat{S}_{A\to C}=\mathcal{P}_{A\to C}e^{i\hat{x}_C(\hat{p}_B-\hat{p})}.
\end{equation}
This closely resembles the transformation in the extended perspectival and perspective-neutral approaches except that the total momentum is now promoted to a quantum operator, giving rise to a different phase depending on the charge sector. Moreover, assuming that Alice and Charlie have access to all observables that commute with the constraint, this means that they also have access to observables that read out the charge sector. This might be implemented, for example, through diagonal projectors of the form $\Pi_{p}=(\boxed{0}, \dots,\boxed{0},\boxed{\mathbb{I}},\boxed{0},\dots,\boxed{0})$, where each block refers to $\mathcal{H}_B^{(p)}\otimes \mathcal{H}_C^{(p)}$ for a specific $p$.

The fact that Alice and Charlie can now directly access the total momentum raises the question whether this allows them to access further global quantities, previously unavailable to them. In particular, one might wonder whether they can access the phase $\varphi$ in the state $\ket{\psi_{\kin}}=\frac{1}{\sqrt{2}}\left(\ket{x_1}_A\ket{z_1}_C+e^{i\varphi}\ket{x_2}_A\ket{z_2}_C\right)\ket{\phi}_B$ relative to Eve. Note that in order to properly capture the extra-particle approach, we have to modify the game slightly, allowing Eve to prepare a superposition or mixture of different total momenta (which is why the discussion in this subsection does not correspond to a \enquote{Level 4}).  In this approach, one would first apply an incoherent $G$-twirl to the resulting kinematical state, thereby obtaining a state that has support on the direct sum of \emph{all} charge sectors. Thus, in order to measure a specific phase $\Phi^{(A)}$, one would have to first measure and postselect on a specific outcome $P$. Then, as at Level 2 and 3, one can obtain the phases
\begin{align}
    \Phi^{(A)}\equiv \varphi + (x_1-x_2)P ,\\
    \Phi^{(C)}\equiv \varphi + (z_1-z_2)P,
\end{align}
which we repeat here for convenience, through full state tomography. However, if Alice and Charlie postselect on a single value of $P$ only, this is not sufficient to separate $\varphi$ from the remaining terms since they have no access to the difference between their absolute positions relative to Eve. By postselecting on different values of $P$, say $P_1$ and $P_2$, Alice could measure $\Phi^{(A)}_{{P_1}}$ and $\Phi^{(A)}_{{P_2}}$, determine
\begin{equation}
    \Phi^{(A)}_{{P_2}}-\Phi^{(A)}_{{P_1}}=(x_1-x_2)(P_1-P_2),
\end{equation}
and from this infer $x_1-x_2$. Note that this constitutes inter-branch information, as it quantifies the difference in $A$'s position between the two branches of the superposition as prepared by Eve. This would give her enough information to determine
\begin{equation}
    \varphi = \Phi^{(A)}_{P_1}-(x_1-x_2)P_1.
\end{equation}
Note, however, that this strategy involves a modification of the game set up at the beginning of this section, in particular one where Eve does not prepare the total state in a single charge sector only (see also Discussion \ref{sec:SuperpositionChargeSectors}). It remains to be investigated to what extent the perspectival and perspective-neutral approaches can be extended to such a scenario (see Discussion Subsection \ref{sec:SuperpositionChargeSectors}) and which additional global properties beyond $P$, such as $\varphi$, would be accessible under the new rules of the game in this case.

\section{Discussion} \label{sec:discussion}

\subsection{Comparison Across Quantum Reference Frame Approaches}

The accessibility analysis of the previous section highlights an important feature of the different approaches to QRFs and their relationships, which has been emphasised in several previous works \cite{Krumm_2021, Hoehn_2021, Carette_2023, Castro_Ruiz_2021, delaHamette_2021_perspectiveneutral,Hoehn_2023_subsystems, devuyst2025relation,Castro-Ruiz_2025}: At the conceptual level, the main difference between the various approaches lies in what is deemed accessible from a given perspective. It is important to keep in mind, though, that our analysis does not compare different QRF approaches as a whole. Instead, we consider access to different sets of observables, which are inspired by those accessible in the various QRF approaches. Throughout our analysis, the states and QRF transformation remain the same -- we focus solely on the extended perspectival/perspective-neutral QRF transformation introduced in Section \ref{sec:extension}.
Nevertheless, our results are reflective of a recent analysis by Castro-Ruiz, Galley, and Loveridge \cite{Castro-Ruiz_2025}, which illustrates the difference between the various QRF approaches in the context of a simple example: a three-partite system with $\mathbb{Z}_2$-symmetry. There, the authors showcase how different assumptions about what one observer can infer about the state relative to the other lead to different QRF transformations. Firstly, assuming that they can only have partial knowledge about one another leads to the \emph{operational approach}. Secondly, restricting the set of global states by imposing zero total charge to the \emph{perspective-neutral and perspectival approach}. Note that this still holds for the extended versions of the perspective-neutral and perspectival approaches developed in this work, since the QRF transformation is still unitary. Thirdly, requiring a single one-to-one map between the perspectives that works across all charge sectors yields the \emph{extra-particle approach}  \cite[pp. 5-6]{Castro-Ruiz_2025}.

Despite employing a different symmetry group, our results can be seen to complement these insights. Rather than starting with different assumptions about what the observers can infer about each other's quantum states, we investigate whether they can infer the total charge (i.e., the total momentum) given access to different sets of observables within their own perspective and considering the possibility of classical communication.

We find that restricting to relative-framed observables (\hyperref[sec: level1]{Level 1}), as in the \emph{operational approach}, prevents the observers from performing full state tomography, making it impossible to infer the total momentum. Conversely, if the total momentum is not accessible to the observers, they cannot uniquely determine how their states transform relative to each other; this justifies the averaging procedure in the QRF transformation used in \cite{Carette_2023,Castro-Ruiz_2025}.  

Allowing access to all relational observables, as in the perspective-neutral approach, or all observables in the relative Hilbert space, as in the perspectival approach (\hyperref[sec: level2]{Level 2}), allows the observers to fully reconstruct the states relative to themselves. Yet, due to their inability to isolate the total momentum, it is only when they can communicate (\hyperref[sec: level3]{Level 3}) that they can infer the total momentum $P$. By allowing for non-zero total momentum, we go beyond the standard perspective-neutral and perspectival formalisms. Yet, the argument of \cite{Castro-Ruiz_2025} remains valid: the restriction to a specific charge sector allows the observers to infer each other's states if they know the total charge -- this is simply the reverse of our reasoning.

Finally, as we briefly discussed at the end (Subsection \ref{sec: extraParticle}), this inference becomes straightforward in the \emph{extra-particle approach}, where the total momentum is encoded in the extra particle and thus directly accessible to each observer without the need for inter-party communication. Again, this aligns with the motivation for the extra-particle approach, which is to find a single map relating states across different reference frames while allowing for arbitrary total charge \cite{Castro_Ruiz_2021,Castro-Ruiz_2025}. Moreover, if one were to modify the initial setup of the game such that Eve can prepare all systems in a superposition or mixture of charge sectors, Alice and Charlie could also infer the relative phase $\varphi$ between the position states of $A$ and $C$ relative to Eve.\\

The results from our accessibility analysis and, in particular, which observables are deemed accessible in the different approaches, are summarised in Table \ref{tab:accesibleObservables} below. 

\begin{table}[h!]
    \centering
    \renewcommand{\arraystretch}{1.3} 
    \setlength{\tabcolsep}{6pt}       
    \begin{tabular}{| l | l | l | l |}
    \hline
         \textbf{Approach} & \textbf{Perspective of Alice} & \textbf{Perspective of Charlie} & \textbf{With communication} \\
    \hline
         Operational & 
         $\hat{x}_B-\hat{x}_A$, $\hat{x}_C-\hat{x}_A$, $\hat{p}_B$ & 
         $\hat{x}_A-\hat{x}_C$, $\hat{x}_B-\hat{x}_C$, $\hat{p}_B$ & 
         rel.\ positions, $\hat{p}_B$ \\
    \hline
         Perspective-neutral & 
         $\hat{x}_B-\hat{x}_A$, $\hat{x}_C-\hat{x}_A$, $\hat{p}_B$, $\hat{p}_C$ & 
         $\hat{x}_A-\hat{x}_C$, $\hat{x}_B-\hat{x}_C$, $\hat{p}_A$, $\hat{p}_B$ & 
         rel.\ positions, $\hat{p}_A$, $\hat{p}_B$, $\hat{p}_C$, $P$ \\
    \hline
         Perspectival & 
         $\hat{x}_B$, $\hat{x}_C$, $\hat{p}_B$, $\hat{p}_C$ & 
         $\hat{x}_A$, $\hat{x}_B$, $\hat{p}_A$, $\hat{p}_B$ & 
         (rel.) positions, $\hat{p}_A$, $\hat{p}_B$, $\hat{p}_C$, $P$ \\
    \hline
         Extra-Particle & 
         $\hat{x}_B-\hat{x}_A$, $\hat{x}_C-\hat{x}_A$, $\hat{p}_B$, $\hat{p}_C$, $P$\textcolor{gray}{, $\varphi$} & 
         $\hat{x}_A-\hat{x}_C$, $\hat{x}_B-\hat{x}_C$, $\hat{p}_A$, $\hat{p}_B$, $P$\textcolor{gray}{, $\varphi$} & 
         rel.\ positions, $\hat{p}_A$, $\hat{p}_B$, $\hat{p}_C$, $P$\textcolor{gray}{, $\varphi$} \\
    \hline
    \end{tabular}
    \caption{Linear combinations of the above observables are accessible from within Alice's and Charlie's perspective in different approaches to QRFs. Note that while represented differently in the table above, the (relative) position operators in the perspective-neutral and perspectival approaches are isomorphic, e.g.\ $\hat{x}_B-\hat{x}_A$ in the perspective-neutral approach corresponds to $\hat{x}_B$ in the perspectival one. Importantly, the total momentum $P$, the eigenvalue of $\hat{p}_A+\hat{p}_B+\hat{p}_C$, is accessible in the perspective-neutral and perspectival approach only if the observers combine their results. In the extra-particle approach, it is directly accessible even without communication. Moreover, allowing for a superposition or mixture of different charge sectors (indicated by grey colour), Alice and Charlie could even infer the phase $\varphi$. It remains to be investigated to what extent the perspectival and perspective-neutral approach can be extended to such a scenario.}
    \label{tab:accesibleObservables}
\end{table}

\subsection{Comment on the Covariance of Physical Laws} \label{sec:CovariancePhysicalLaws}

Our results also relate to recent work by Mekonnen, Galley and Müller \cite{Mekonnen_2025}, which, inter alia, discusses the possibility of using different representations of the translation group, depending on the total momentum $P$. This leads to the same transformation as our generalised QRF transformation \eqref{eq: genQRFtrafo}. They then apply this transformation to the setup considered in \cite{delaHamette_2021_indefinitemetric}: a test particle in the presence of a massive object in superposition. In this work, we used the principle of \emph{covariance of physical laws under QRF transformations} (cf.~\cite{Giacomini_2017_covariance}) to argue that one can use a QRF transformation to go into the reference frame of the massive object, in which it will assume a definite position and thus source a definite gravitational field, solve the problem in this reference frame, and then use the inverse QRF transformation to infer the motion of the test particle in the original frame of reference. A key ingredient in this argument was the \emph{form-invariance of the Hamiltonian}. Neglecting the momentum of the massive object as well as the original QRF $R$, the Hamiltonian in the respective perspectives is given by
\begin{align}
        \hat{H}_{RS}^{(M)} &=  \frac{\hat{\vec{\pi}}_S^2}{2m_S} + m_S \hat{V}(\hat{\vec{q}}_S),\label{eq:HamiltonianM}\\
		\hat{H}_{MS}^{(R)} &= \frac{\hat{\vec{p}}_S^2}{2m_S} + m_S V(\hat{\vec{x}}_S-\hat{\vec{x}}_M). \label{eq:HamiltonianR}
\end{align}
These operators are invariant under QRF transformations, with the caveat that the position of $R$ is set to zero in its own frame (see also the end of Section \ref{sec:invariantobservables}). Importantly, $\hat{H}^{(M)}_{RS}$ transforms in precisely the same way under the generalised QRF transformation $\hat{S}_{M\to R}^P=\mathcal{P}_{MR}e^{i\hat{\vec{q}}_R(\hat{\vec{\pi}}_S-P)}$ as under the standard one, $\hat{S}_{M\to R}^0$, since the Hamiltonian commutes with $\hat{q}_R$. That is, the argument extends unproblematically to the non-zero charge sector. This is in contrast with the analysis of \cite{Mekonnen_2025}, who argue that \enquote{physics is only invariant under a single representation of the quantum symmetry group and not under all QRF transformations}. How can one explain this difference? While the authors of \cite{Mekonnen_2025} do not explicitly state how they compute the evolution, they drop the assumption that the reference system $R$ is far away enough such that it is not affected by the gravitational field of $M$. This is a crucial assumption in \cite{delaHamette_2021_indefinitemetric}, which justifies neglecting the momentum $\hat{p}_R$ in the Hamiltonian relative to $M$ (Eq.~\eqref{eq:HamiltonianM}). Thus, we assume that the dynamical evolution in \cite{Mekonnen_2025} must have been computed using the Hamiltonian
\begin{align}
   H^{(M)}_{SR}=\frac{\hat{\pi}_R^2}{2m_R}+\frac{\hat{\pi}_S^2}{2m_S}+m_R V(\hat{q}_R)+m_SV(\hat{q}_S),
\end{align}
which is \emph{not} invariant under $\hat{S}^{P}_{M\to R}$ \emph{for any value of P, including $P=0$}, due to the presence of the term proportional to $\hat{\pi}_R$. As a consequence, $\hat{S}^{P}_{M\to R}$ does not constitute a symmetry in this case and the strategy of computing the dynamical evolution in a more tractable QRF is not applicable -- independently of the charge sector.\footnote{Since the authors of \cite{Mekonnen_2025} do not explicitly state the Hamiltonian generating the time evolution (employing a semi-classical approximation to compute the trajectory of the probe $S$ in each branch instead), let us note that this argument applies to essentially all Hamiltonians $\hat{H}_{RS}^{(M)}$ which depend on $\hat{\pi}_R$ -- it would require extreme fine-tuning to obtain a Hamiltonian depending on $\hat{\pi}_R$ that is invariant under $\exp(i\hat{q}_R\hat{p}_S)$ while it fails to be invariant under $\exp(i\hat{q}_R P)$.}

\subsection{Further Generalisation of the QRF Transformation}

One may ask whether the extended QRF transformation we considered here could be generalised even further, as suggested in \cite{Castro-Ruiz_2025}. A first natural modification of the transformation $\hat{S}_{A\to C}^P = \hat{\mathcal{P}}_{AC}e^{i\hat{x}_C(\hat{p}_B-P)}$ would be to replace the term $\hat{x}_CP$ in the exponent by an arbitrary function $f(\hat{x}_C)$ of the position operator, resulting in a QRF transformation 
\begin{equation} \label{eq:genQRFoperator}
\hat{S}^f_{AC}= \hat{\mathcal{P}}_{AC}e^{i\hat{x}_C\hat{p}_B}e^{if(\hat{x}_C)}
\end{equation}
with a more general, branch-dependent phase. Note that these maps preserve the transformation properties of the position variables in Eq.~\eqref{eq:1a}--\eqref{eq:1b} and can thus be understood as a form of quantum translation, albeit with modified transformation properties of the momentum $\hat{p}_C$. Let us briefly show that the transformations considered in this paper are already of the most general form compatible with the basic properties of frame changes. One necessary condition we require the most generic QRF transformation operator to satisfy is transitivity, that is, 
\begin{align}
    \hat{S}^f_{2\to 3}\hat{S}^f_{1\to 2}=\hat{S}^f_{1\to 3}.
\end{align}
We find that 
\begin{align}
    \hat{S}^f_{2\to 3}\hat{S}^f_{1\to 2} \int dx dy \psi(x,y) \ket{x}_2\ket{y}_3 &= \hat{S}^f_{2\to 3} \hat{\mathcal{P}}_{12} \int dx dy \psi(x,y) \ket{x}_2\ket{y-x}_3 e^{if(x)} \\
    &= \hat{S}^f_{2\to 3} \int dxdy \psi(x,y) \ket{-x}_1\ket{y-x}_3 e^{if(x)} \\
    &= \hat{\mathcal{P}}_{23} \int dxdy \psi(x,y) \ket{-x-(y-x)}_1\ket{y-x}_3 e^{if(y-x)}e^{if(x)} \\
    &= \int dxdy \psi(x,y) \ket{-y}_1\ket{x-y}_3 e^{if(y-x)}e^{if(x)},
\end{align}
which equals $\hat{S}^f_{1\to 3} \int dx dy \psi(x,y) \ket{x}_2\ket{y}_3$ if and only if 
\begin{align}
    e^{if(y-x)}e^{if(x)}=e^{if(y)}.
\end{align}
Moreover, imposing unitarity, that is, $(\hat{S}^f_{1\to 2})^\dagger = \hat{S}^f_{2\to 1}$, implies that $f$ is an odd function. As a consequence, $f(0)=0$ and we have shown that $f$ must be a linear function.\\

In general, the function $f$ in Eq.~\eqref{eq:genQRFoperator} could depend on all phase space variables of systems $B$ and $C$ (when transforming from system $A$), that is, $f(\hat{x}_C, \hat{p}_C, \hat{x}_B, \hat{p}_B)$. We provide additional arguments in Appendix \ref{app:mostgeneralQRFtrafo} for why this cannot be the case. We thus conclude that the form $\hat{\mathcal{P}}_{AC}e^{i\hat{x}_C(\hat{p}_B-c)}$ for $c\in \mathbb{R}$ is the most generic form that a transitive and unitary QRF change operator can take. \\

\subsection{Superposition of Charge Sectors} \label{sec:SuperpositionChargeSectors}

Another natural question that arises when extending QRF transformations to arbitrary charge sectors is whether we can also consider a \emph{superposition of charge sectors}. The first challenge when trying to address this question is the presence of superselection rules \cite{Kitaev_2004, Bartlett_2007, Palmer_2014}. In the presence of a superselection rule associated with strict translation invariance, it is not possible to consider a superposition of total momenta, since such a superposition cannot be distinguished from a mixture. This is because verifying coherence between different momentum sectors would require measurements in the conjugate basis, i.e.~of the centre of mass position operator, which are observables inaccessible to internal subsystems under translation invariance. Yet, we know that it is, in principle, possible to prepare and measure systems that move in a superposition of different total momenta. This is because \emph{relative} to an external reference frame, such as the laboratory, one can make sense of the \enquote{absolute} positions of the systems. Thus, the superselection rule doesn't apply \emph{at the level of the laboratory}.

So, given that it makes sense to consider a superposition of charge sectors relative to an external frame, is there a way of extending our construction to such scenarios? At the end of Section \ref{sec:extensionPN}, we saw that we can derive the QRF transformation for non-zero total momentum $P$ from the one for zero charge sector by including a fourth system $D$ that moves at momentum $-P$ relative to Eve. One might hope to extend this strategy by preparing $D$ in a superposition of momenta. To understand whether this may work, let us consider $D$ in a superposition of momenta $-P_1$ and $-P_2$, such that the kinematical state is given by 
\begin{equation}
    \ket{\psi_{\kin}}^{(E)}_{ABCD} = \frac{1}{\sqrt{2}}(\ket{x_1}_A+\ket{x_2}_A)\ket{y}_B\ket{z}_C\frac{1}{\sqrt{2}}\left(\ket{-P_1}_D+\ket{-P_2}_D\right).
\end{equation}
Applying the coherent $G$-twirl on the four-party system and going into the perspective of $A$ (with $X=0$) gives
\begin{align}
    \ket{\psi}_{BCD}^{(A)}=\frac{1}{\sqrt{2}}&\left(\ket{\psi^{P_1}}^{(A)}_{BC}\ket{-P_1}_D+\ket{\psi^{P_2}}^{(A)}_{BC}\ket{-P_2}_D\right)\\
    =\frac{1}{\sqrt{2}}&(e^{-ix_1P_1}\ket{y-x_1}_B\ket{z-x_1}_C+e^{-ix_2P_1}\ket{y-x_2}_B\ket{z-x_2}_C)\ket{-P_1}_D\\
  + &(e^{-ix_1P_2}\ket{y-x_1}_B\ket{z-x_1}_C+e^{-ix_2P_2}\ket{y-x_2}_B\ket{z-x_2}_C)\ket{-P_2}_D.
\end{align}
This, however, leads to an entangled state, which, upon tracing out the additional system $D$, leads to a mixture of different charge sectors (cf.~\cite{Bartlett_2007,Krumm_2021}) with respect to Alice, in line with the superselection rule. Thus, while one can in principle prepare a superposition of different charge sectors relative to an external frame (e.g.~Eve), one (e.g.~Alice) \emph{cannot} distinguish it from a mixture \emph{from within}, i.e.~without access to an external reference frame or the additional system $D$. Of course, Alice is free to prepare and measure systems in momentum superpositions relative to her own reference frame as the perspectival approach places no restrictions on such operations. 

Note that this strategy of adding a subsystem $D$ differs from the one employed in the extra-particle approach, which consists of encoding the total momentum information in $\H_p\simeq \mathbb{C}$ \cite{Castro_Ruiz_2021}. Unlike system $D$, this Hilbert space does not capture both position and momentum variables but only encodes the value of the total momentum, i.e.~the charge sector.\footnote{Mathematically, this implies that we no longer work with a symplectic structure, as usual for the phase space, but a symplectic \emph{foliation}, which makes standard quantisation difficult (see e.g.~\cite{Riello_2021}).}\footnote{Let us also briefly note that a superposition of charge sectors is also briefly discussed in the context of gauge theories in \cite{Araujoregado_2025} -- although in this context, one needs to distinguish between the subregion of interest, whose boundary carries zero charge, and the global region, in which the subregion is embedded and which can be considered to carry non-trivial charges at its boundary.} \\

So, to summarise, one can, in the presence of translation invariance, always prepare and measure superpositions of momenta with respect to one’s own reference frame, but one cannot verify that oneself is in a superposition of charge sectors. In the present case, we {\it define} that Eve possesses an external reference frame with respect to which she can prepare superpositions of momenta. However, if, at the same time, Eve were part of a \emph{larger} translationally invariant system as seen from yet another external reference, she would again not be able to find out whether she is in a superposition of momenta with respect to this reference frame. In that sense, her role would be comparable to that of Alice, just on a higher level.\footnote{It would be interesting to investigate to what extent this view stands in contrast to findings in \cite{delsanto2025}.}

\section{Conclusion} \label{sec:conclusion}

In this work, we provided an extension of the perspectival and perspective-neutral approach to different charge sectors, considering the possibility that the entire system, including the reference frames, moves at a fixed total momentum $P$ (relative to an external reference frame) (Section \ref{sec:extension}). In the perspectival approach, this is achieved by altering the transformation of the canonical variables. In the perspective-neutral approach, we modified the constraint, which is imposed on the kinematical states to obtain the physical Hilbert space. The resulting QRF transformations reduce to the standard transformations upon setting $P=0$ and otherwise differ from them by a phase, which depends on both the total momentum and the position of the new reference frame relative to the old. While this phase is global and thus unobservable in the case of classical reference frames with well-defined positions relative to one another, it becomes a relative phase and thus in principle observable when working with quantum reference frames. To understand precisely under which conditions two internal observers can measure this phase and infer the total momentum, we undertook an accessibility analysis in Section \ref{sec:accessibility}. Gradually increasing the observables and resources available to them, we studied three different levels of accessibility, inspired by different approaches to QRFs. At the first level, Alice and Charlie only have access to \emph{relative framed observables} as used in the operational approach \cite{Carette_2023} to QRFs. In this case, they cannot determine the phase at all. At the second level, they can measure all \emph{relational observables} of the perspective-neutral approach, which correspond to \emph{all observables in a given perspective} of the perspectival approach. While this allows them to perform full state tomography and thus determine the phase, it is only at the third level, which also allows for classical communication, that they can combine their results to infer the total momentum. Finally, going beyond these three levels, we discussed how the extra-particle approach allows direct access to the total momentum and might thus give access to additional quantities encoded in the entangled state of $A$ and $C$ relative to Eve. 

Beyond the accessibility analysis, our theoretical framework contributes several new insights to the field of QRFs. Our invariant observables Theorem \ref{th: invariance} characterises precisely which operators remain invariant under the extended transformations, revealing that the additional momentum-dependent phase imposes stricter invariance conditions than the standard zero charge case. More fundamentally, we demonstrated that the extended QRF transformation presented here represents the most general form of a quantum-controlled translation compatible with transitivity and unitarity, establishing it as a natural completion of the standard framework rather than an ad hoc extension. Additionally, our analysis in Section \ref{sec:CovariancePhysicalLaws} confirms that the principle of covariance of physical laws -- an essential feature of the perspectival approach -- continues to hold in the extended framework, provided the underlying dynamics commute appropriately with the reference frame positions.

Our results can be embedded in a larger research programme, trying to understand the relation between different approaches to QRFs \cite{Krumm_2021, Hoehn_2021, Carette_2023, Castro_Ruiz_2021, delaHamette_2021_perspectiveneutral,Hoehn_2023_subsystems, devuyst2025relation,Castro-Ruiz_2025}. In line with previous arguments, our analysis highlights that what is deemed accessible in a given perspective strongly depends on the chosen framework. As such, the question whether the total momentum can be accessed \emph{from within} has no universal answer. Different approaches to QRFs will lead to different conclusions. This also suggests that there is no one right approach to modelling QRFs but rather that different setups require different theoretical frameworks.\\

While our analysis is based on states that are, by construction, compatible with a global perspective, one may ask whether there are states that do not give rise to a global perspective at all. In this respect, the uniqueness results regarding the extended QRF transformation may prove to be helpful. We saw that the QRF transformation presented in this work is the most general position reference frame transformation compatible with transitivity and unitarity. Conversely, if we start from two states that are not related by such a transformation, it might be impossible to find a reasonable change of reference frame that treats these states as different perspectives in the QRF sense. In certain cases, one might be able to cast them as relative to \emph{non-ideal} reference frames, though we leave the extension of QRF transformations for non-ideal frames to non-zero charge sectors for future work. However, we would like to highlight that a particularly interesting case would be non-ideal QRFs with a restricted momentum spectrum -- here, one might want to investigate how the relation between the different spectra changes once the total momentum is altered. Another interesting further direction would be to extend the present framework to field-theories, such as $U(1)$ gauge theory. The extension to non-abelian theories presents a bigger challenge since there the zero charge sector is, in a sense, preferred as it is the only sector in which the states can be considered to be invariant.

Finally, let us return once more to the question of whether two arbitrary quantum states can be embedded in a global perspective. Even if one takes into account non-ideal QRFs or generalises the construction to other symmetry groups, some states might not be relatable by \emph{any} reasonable QRF transformation. This, combined with the limitations in reconstructing the global perspective from the local ones, that was explicated in our accessibility analysis, points towards an asymmetry in the relationship between local and global descriptions. While global perspectives always generate consistent local descriptions through appropriate reduction maps, the converse -- inferring global properties from local perspectives -- proves far more subtle and framework-dependent. This asymmetry also shows up beyond the context of QRF transformations, in particular in recent thought experiments, such as the Wigner's friend scenario and its extensions (e.g~\cite{Wigner_1962,Brukner_2017,Frauchiger_2018, Bong_2019}). Such scenarios increasingly suggest that the very notion of a universal, objective perspective may itself be called into question. Understanding what can be reconstructed from within a perspective and what may remain ultimately inaccessible thus becomes increasingly important as we push the boundaries of quantum theory beyond the well-established regimes.\\

\textit{Note:} In the final stages of this paper, we became aware of independent work by Guilhem Doat and Augustin Vanrietvelde, who identify charge accessibility as the source of disagreement between quantum reference frames frameworks and operationally motivate this accessibility. This work will appear on the arXiv shortly.

\begin{acknowledgments}
All authors thank Esteban Castro and Bruna Sahdo for helpful discussions regarding the extra-particle approach. AC thanks Andrea Di Biagio for insightful discussions. VK thanks the QIT group at ETH, in particular Ladina Hausmann and Sébastien Garmier for helpful discussions and feedback. The authors also acknowledge the support of ChatGPT 5 in the development of the proofs of Theorem 1 and Proposition 1, the conception of the counterexample, as well as the proof in Appendix E.

ACdlH acknowledges support from the Institute for Theoretical Studies at ETH Zurich through a Junior Research Fellowship. VK thanks the Swiss National Science Foundation via the NCCR SwissMAP and the ETH Zurich Quantum Center for support. 
This research was funded in whole or in part by the Austrian Science Fund (FWF) [10.55776/F71] and [10.55776/COE1]. For open access purposes, the authors have applied a CC BY public copyright license to any author accepted manuscript version arising from this submission. Funded by the European Union - NextGenerationEU. 
This publication was made possible through the financial support of the ID 62312 grant from the John Templeton Foundation, as part of The Quantum Information Structure of Spacetime (QISS) Project (qiss.fr). The opinions expressed in this publication are those of the authors and do not necessarily reflect the views of the John Templeton Foundation.
\end{acknowledgments}

\nocite{apsrev42Control}
\bibliography{bibliography}

\begin{thebibliography}{58}%
\makeatletter
\providecommand \@ifxundefined [1]{%
 \@ifx{#1\undefined}
}%
\providecommand \@ifnum [1]{%
 \ifnum #1\expandafter \@firstoftwo
 \else \expandafter \@secondoftwo
 \fi
}%
\providecommand \@ifx [1]{%
 \ifx #1\expandafter \@firstoftwo
 \else \expandafter \@secondoftwo
 \fi
}%
\providecommand \natexlab [1]{#1}%
\providecommand \enquote  [1]{``#1''}%
\providecommand \bibnamefont  [1]{#1}%
\providecommand \bibfnamefont [1]{#1}%
\providecommand \citenamefont [1]{#1}%
\providecommand \href@noop [0]{\@secondoftwo}%
\providecommand \href [0]{\begingroup \@sanitize@url \@href}%
\providecommand \@href[1]{\@@startlink{#1}\@@href}%
\providecommand \@@href[1]{\endgroup#1\@@endlink}%
\providecommand \@sanitize@url [0]{\catcode `\\12\catcode `\$12\catcode `\&12\catcode `\#12\catcode `\^12\catcode `\_12\catcode `\%12\relax}%
\providecommand \@@startlink[1]{}%
\providecommand \@@endlink[0]{}%
\providecommand \url  [0]{\begingroup\@sanitize@url \@url }%
\providecommand \@url [1]{\endgroup\@href {#1}{\urlprefix }}%
\providecommand \urlprefix  [0]{URL }%
\providecommand \Eprint [0]{\href }%
\providecommand \doibase [0]{https://doi.org/}%
\providecommand \selectlanguage [0]{\@gobble}%
\providecommand \bibinfo  [0]{\@secondoftwo}%
\providecommand \bibfield  [0]{\@secondoftwo}%
\providecommand \translation [1]{[#1]}%
\providecommand \BibitemOpen [0]{}%
\providecommand \bibitemStop [0]{}%
\providecommand \bibitemNoStop [0]{.\EOS\space}%
\providecommand \EOS [0]{\spacefactor3000\relax}%
\providecommand \BibitemShut  [1]{\csname bibitem#1\endcsname}%
\let\auto@bib@innerbib\@empty
\bibitem [{\citenamefont {{Giacomini}}\ \emph {et~al.}(2019)\citenamefont {{Giacomini}}, \citenamefont {{Castro-Ruiz}},\ and\ \citenamefont {{Brukner}}}]{Giacomini_2017_covariance}%
  \BibitemOpen
  \bibfield  {author} {\bibinfo {author} {\bibfnamefont {F.}~\bibnamefont {{Giacomini}}}, \bibinfo {author} {\bibfnamefont {E.}~\bibnamefont {{Castro-Ruiz}}},\ and\ \bibinfo {author} {\bibfnamefont {{\v{C}}.}~\bibnamefont {{Brukner}}},\ }\bibfield  {title} {\bibinfo {title} {{Quantum mechanics and the covariance of physical laws in quantum reference frames}},\ }\href {https://doi.org/10.1038/s41467-018-08155-0} {\bibfield  {journal} {\bibinfo  {journal} {Nature Communications}\ }\textbf {\bibinfo {volume} {10}},\ \bibinfo {eid} {494} (\bibinfo {year} {2019})}\BibitemShut {NoStop}%
\bibitem [{\citenamefont {Giacomini}\ \emph {et~al.}(2019)\citenamefont {Giacomini}, \citenamefont {Castro-Ruiz},\ and\ \citenamefont {Brukner}}]{Giacomini_2018_spin}%
  \BibitemOpen
  \bibfield  {author} {\bibinfo {author} {\bibfnamefont {F.}~\bibnamefont {Giacomini}}, \bibinfo {author} {\bibfnamefont {E.}~\bibnamefont {Castro-Ruiz}},\ and\ \bibinfo {author} {\bibfnamefont {{\v{C}}.}~\bibnamefont {Brukner}},\ }\bibfield  {title} {\bibinfo {title} {Relativistic quantum reference frames: The operational meaning of spin},\ }\href {https://doi.org/10.1103/PhysRevLett.123.090404} {\bibfield  {journal} {\bibinfo  {journal} {Phys. Rev. Lett.}\ }\textbf {\bibinfo {volume} {123}},\ \bibinfo {pages} {090404} (\bibinfo {year} {2019})}\BibitemShut {NoStop}%
\bibitem [{\citenamefont {de~la Hamette}\ and\ \citenamefont {Galley}(2020)}]{delaHamette_2020}%
  \BibitemOpen
  \bibfield  {author} {\bibinfo {author} {\bibfnamefont {A.-C.}\ \bibnamefont {de~la Hamette}}\ and\ \bibinfo {author} {\bibfnamefont {T.~D.}\ \bibnamefont {Galley}},\ }\bibfield  {title} {\bibinfo {title} {Quantum reference frames for general symmetry groups},\ }\href {https://doi.org/10.22331/q-2020-11-30-367} {\bibfield  {journal} {\bibinfo  {journal} {{Quantum}}\ }\textbf {\bibinfo {volume} {4}},\ \bibinfo {pages} {367} (\bibinfo {year} {2020})}\BibitemShut {NoStop}%
\bibitem [{\citenamefont {Mikusch}\ \emph {et~al.}(2021)\citenamefont {Mikusch}, \citenamefont {Barbado},\ and\ \citenamefont {Brukner}}]{Mikusch_2021}%
  \BibitemOpen
  \bibfield  {author} {\bibinfo {author} {\bibfnamefont {M.}~\bibnamefont {Mikusch}}, \bibinfo {author} {\bibfnamefont {L.~C.}\ \bibnamefont {Barbado}},\ and\ \bibinfo {author} {\bibfnamefont {{\v{C}}.}~\bibnamefont {Brukner}},\ }\bibfield  {title} {\bibinfo {title} {Transformation of spin in quantum reference frames},\ }\href {https://doi.org/10.1103/PhysRevResearch.3.043138} {\bibfield  {journal} {\bibinfo  {journal} {Phys. Rev. Research}\ }\textbf {\bibinfo {volume} {3}},\ \bibinfo {pages} {043138} (\bibinfo {year} {2021})}\BibitemShut {NoStop}%
\bibitem [{\citenamefont {Barbado}\ \emph {et~al.}(2020)\citenamefont {Barbado}, \citenamefont {Castro-Ruiz}, \citenamefont {Apadula},\ and\ \citenamefont {Brukner}}]{Barbado_2020}%
  \BibitemOpen
  \bibfield  {author} {\bibinfo {author} {\bibfnamefont {L.~C.}\ \bibnamefont {Barbado}}, \bibinfo {author} {\bibfnamefont {E.}~\bibnamefont {Castro-Ruiz}}, \bibinfo {author} {\bibfnamefont {L.}~\bibnamefont {Apadula}},\ and\ \bibinfo {author} {\bibfnamefont {{\v{C}}.}~\bibnamefont {Brukner}},\ }\bibfield  {title} {\bibinfo {title} {Unruh effect for detectors in superposition of accelerations},\ }\href {https://doi.org/10.1103/PhysRevD.102.045002} {\bibfield  {journal} {\bibinfo  {journal} {Phys. Rev. D}\ }\textbf {\bibinfo {volume} {102}},\ \bibinfo {pages} {045002} (\bibinfo {year} {2020})}\BibitemShut {NoStop}%
\bibitem [{\citenamefont {Apadula}\ \emph {et~al.}(2024)\citenamefont {Apadula}, \citenamefont {Castro-Ruiz},\ and\ \citenamefont {Brukner}}]{Apadula_2022}%
  \BibitemOpen
  \bibfield  {author} {\bibinfo {author} {\bibfnamefont {L.}~\bibnamefont {Apadula}}, \bibinfo {author} {\bibfnamefont {E.}~\bibnamefont {Castro-Ruiz}},\ and\ \bibinfo {author} {\bibfnamefont {{\v{C}}.}~\bibnamefont {Brukner}},\ }\bibfield  {title} {\bibinfo {title} {Quantum {R}eference {F}rames for {L}orentz {S}ymmetry},\ }\href {https://doi.org/10.22331/q-2024-08-14-1440} {\bibfield  {journal} {\bibinfo  {journal} {{Quantum}}\ }\textbf {\bibinfo {volume} {8}},\ \bibinfo {pages} {1440} (\bibinfo {year} {2024})}\BibitemShut {NoStop}%
\bibitem [{\citenamefont {Giacomini}\ and\ \citenamefont {Brukner}(2020)}]{Giacomini_2020_einstein}%
  \BibitemOpen
  \bibfield  {author} {\bibinfo {author} {\bibfnamefont {F.}~\bibnamefont {Giacomini}}\ and\ \bibinfo {author} {\bibfnamefont {{\v{C}}.}~\bibnamefont {Brukner}},\ }\href@noop {} {\bibinfo {title} {{Einstein's Equivalence principle for superpositions of gravitational fields and quantum reference frames}}} (\bibinfo {year} {2020}),\ \Eprint {https://arxiv.org/abs/2012.13754} {arXiv:2012.13754 [quant-ph]} \BibitemShut {NoStop}%
\bibitem [{\citenamefont {Giacomini}(2021)}]{Giacomini_2021}%
  \BibitemOpen
  \bibfield  {author} {\bibinfo {author} {\bibfnamefont {F.}~\bibnamefont {Giacomini}},\ }\bibfield  {title} {\bibinfo {title} {Spacetime quantum reference frames and superpositions of proper times},\ }\href {https://doi.org/10.22331/q-2021-07-22-508} {\bibfield  {journal} {\bibinfo  {journal} {Quantum}\ }\textbf {\bibinfo {volume} {5}},\ \bibinfo {pages} {508} (\bibinfo {year} {2021})}\BibitemShut {NoStop}%
\bibitem [{\citenamefont {Giacomini}\ and\ \citenamefont {Brukner}(2022)}]{Giacomini_2021_einstein}%
  \BibitemOpen
  \bibfield  {author} {\bibinfo {author} {\bibfnamefont {F.}~\bibnamefont {Giacomini}}\ and\ \bibinfo {author} {\bibfnamefont {{\v{C}}.}~\bibnamefont {Brukner}},\ }\bibfield  {title} {\bibinfo {title} {{Quantum superposition of spacetimes obeys {E}instein's {E}quivalence {P}rinciple}},\ }\href {https://doi.org/10.1116/5.0070018} {\bibfield  {journal} {\bibinfo  {journal} {AVS Quantum Sci.}\ }\textbf {\bibinfo {volume} {4}},\ \bibinfo {pages} {015601} (\bibinfo {year} {2022})}\BibitemShut {NoStop}%
\bibitem [{\citenamefont {Cepollaro}\ and\ \citenamefont {Giacomini}(2024)}]{Cepollaro_2021}%
  \BibitemOpen
  \bibfield  {author} {\bibinfo {author} {\bibfnamefont {C.}~\bibnamefont {Cepollaro}}\ and\ \bibinfo {author} {\bibfnamefont {F.}~\bibnamefont {Giacomini}},\ }\bibfield  {title} {\bibinfo {title} {{Quantum generalisation of Einstein’s equivalence principle can be verified with entangled clocks as quantum reference frames}},\ }\href {https://doi.org/10.1088/1361-6382/ad6d26} {\bibfield  {journal} {\bibinfo  {journal} {Classical and Quantum Gravity}\ }\textbf {\bibinfo {volume} {41}},\ \bibinfo {pages} {185009} (\bibinfo {year} {2024})}\BibitemShut {NoStop}%
\bibitem [{\citenamefont {Kabel}\ \emph {et~al.}(2024)\citenamefont {Kabel}, \citenamefont {de~la Hamette}, \citenamefont {Castro-Ruiz},\ and\ \citenamefont {Brukner}}]{Kabel_2022_conformal}%
  \BibitemOpen
  \bibfield  {author} {\bibinfo {author} {\bibfnamefont {V.}~\bibnamefont {Kabel}}, \bibinfo {author} {\bibfnamefont {A.-C.}\ \bibnamefont {de~la Hamette}}, \bibinfo {author} {\bibfnamefont {E.}~\bibnamefont {Castro-Ruiz}},\ and\ \bibinfo {author} {\bibfnamefont {{\v{C}}.}~\bibnamefont {Brukner}},\ }\bibfield  {title} {\bibinfo {title} {Quantum conformal symmetries for spacetimes in superposition},\ }\href {https://doi.org/10.22331/q-2024-12-04-1547} {\bibfield  {journal} {\bibinfo  {journal} {{Quantum}}\ }\textbf {\bibinfo {volume} {8}},\ \bibinfo {pages} {1547} (\bibinfo {year} {2024})}\BibitemShut {NoStop}%
\bibitem [{\citenamefont {de~la Hamette}\ \emph {et~al.}(2025)\citenamefont {de~la Hamette}, \citenamefont {Kabel}, \citenamefont {Christodoulou},\ and\ \citenamefont {Brukner}}]{delahamette_2022_DiffsICO}%
  \BibitemOpen
  \bibfield  {author} {\bibinfo {author} {\bibfnamefont {A.-C.}\ \bibnamefont {de~la Hamette}}, \bibinfo {author} {\bibfnamefont {V.}~\bibnamefont {Kabel}}, \bibinfo {author} {\bibfnamefont {M.}~\bibnamefont {Christodoulou}},\ and\ \bibinfo {author} {\bibfnamefont {{\v{C}}.}~\bibnamefont {Brukner}},\ }\bibfield  {title} {\bibinfo {title} {Indefinite causal order and quantum coordinates},\ }\href {https://doi.org/10.1103/bnkn-4p3f} {\bibfield  {journal} {\bibinfo  {journal} {Phys. Rev. Lett.}\ }\textbf {\bibinfo {volume} {135}},\ \bibinfo {pages} {141402} (\bibinfo {year} {2025})}\BibitemShut {NoStop}%
\bibitem [{\citenamefont {Wang}\ \emph {et~al.}(2024)\citenamefont {Wang}, \citenamefont {Giacomini}, \citenamefont {Nori},\ and\ \citenamefont {Blencowe}}]{Wang_2023}%
  \BibitemOpen
  \bibfield  {author} {\bibinfo {author} {\bibfnamefont {H.}~\bibnamefont {Wang}}, \bibinfo {author} {\bibfnamefont {F.}~\bibnamefont {Giacomini}}, \bibinfo {author} {\bibfnamefont {F.}~\bibnamefont {Nori}},\ and\ \bibinfo {author} {\bibfnamefont {M.~P.}\ \bibnamefont {Blencowe}},\ }\bibfield  {title} {\bibinfo {title} {Relational superposition measurements with a material quantum ruler},\ }\href {https://doi.org/10.22331/q-2024-05-06-1335} {\bibfield  {journal} {\bibinfo  {journal} {{Quantum}}\ }\textbf {\bibinfo {volume} {8}},\ \bibinfo {pages} {1335} (\bibinfo {year} {2024})}\BibitemShut {NoStop}%
\bibitem [{\citenamefont {Kabel}\ \emph {et~al.}(2025)\citenamefont {Kabel}, \citenamefont {de~la Hamette}, \citenamefont {Apadula}, \citenamefont {Cepollaro}, \citenamefont {Gomes}, \citenamefont {Butterfield},\ and\ \citenamefont {Brukner}}]{Kabel_2024}%
  \BibitemOpen
  \bibfield  {author} {\bibinfo {author} {\bibfnamefont {V.}~\bibnamefont {Kabel}}, \bibinfo {author} {\bibfnamefont {A.-C.}\ \bibnamefont {de~la Hamette}}, \bibinfo {author} {\bibfnamefont {L.}~\bibnamefont {Apadula}}, \bibinfo {author} {\bibfnamefont {C.}~\bibnamefont {Cepollaro}}, \bibinfo {author} {\bibfnamefont {H.}~\bibnamefont {Gomes}}, \bibinfo {author} {\bibfnamefont {J.}~\bibnamefont {Butterfield}},\ and\ \bibinfo {author} {\bibfnamefont {{\v C}.}~\bibnamefont {Brukner}},\ }\bibfield  {title} {\bibinfo {title} {Quantum coordinates, localisation of events, and the quantum hole argument},\ }\href {https://doi.org/10.1038/s42005-025-02084-3} {\bibfield  {journal} {\bibinfo  {journal} {Communications Physics}\ }\textbf {\bibinfo {volume} {8}},\ \bibinfo {pages} {185} (\bibinfo {year} {2025})}\BibitemShut {NoStop}%
\bibitem [{\citenamefont {Cepollaro}\ \emph {et~al.}(2025)\citenamefont {Cepollaro}, \citenamefont {Akil}, \citenamefont {Cie\ifmmode \acute{s}\else \'{s}\fi{}li\ifmmode~\acute{n}\else \'{n}\fi{}ski}, \citenamefont {de~la Hamette},\ and\ \citenamefont {Brukner}}]{cepollaro2024}%
  \BibitemOpen
  \bibfield  {author} {\bibinfo {author} {\bibfnamefont {C.}~\bibnamefont {Cepollaro}}, \bibinfo {author} {\bibfnamefont {A.}~\bibnamefont {Akil}}, \bibinfo {author} {\bibfnamefont {P.}~\bibnamefont {Cie\ifmmode \acute{s}\else \'{s}\fi{}li\ifmmode~\acute{n}\else \'{n}\fi{}ski}}, \bibinfo {author} {\bibfnamefont {A.-C.}\ \bibnamefont {de~la Hamette}},\ and\ \bibinfo {author} {\bibfnamefont {{\v{C}}.}~\bibnamefont {Brukner}},\ }\bibfield  {title} {\bibinfo {title} {Sum of entanglement and subsystem coherence is invariant under quantum reference frame transformations},\ }\href {https://doi.org/10.1103/h6b3-y4vt} {\bibfield  {journal} {\bibinfo  {journal} {Phys. Rev. Lett.}\ }\textbf {\bibinfo {volume} {135}},\ \bibinfo {pages} {010201} (\bibinfo {year} {2025})}\BibitemShut {NoStop}%
\bibitem [{\citenamefont {Vanrietvelde}\ \emph {et~al.}(2020)\citenamefont {Vanrietvelde}, \citenamefont {Höhn}, \citenamefont {Giacomini},\ and\ \citenamefont {Castro-Ruiz}}]{Vanrietvelde_2018a}%
  \BibitemOpen
  \bibfield  {author} {\bibinfo {author} {\bibfnamefont {A.}~\bibnamefont {Vanrietvelde}}, \bibinfo {author} {\bibfnamefont {P.~A.}\ \bibnamefont {Höhn}}, \bibinfo {author} {\bibfnamefont {F.}~\bibnamefont {Giacomini}},\ and\ \bibinfo {author} {\bibfnamefont {E.}~\bibnamefont {Castro-Ruiz}},\ }\bibfield  {title} {\bibinfo {title} {A change of perspective: switching quantum reference frames via a perspective-neutral framework},\ }\href {https://doi.org/10.22331/q-2020-01-27-225} {\bibfield  {journal} {\bibinfo  {journal} {Quantum}\ }\textbf {\bibinfo {volume} {4}},\ \bibinfo {pages} {225} (\bibinfo {year} {2020})}\BibitemShut {NoStop}%
\bibitem [{\citenamefont {Vanrietvelde}\ \emph {et~al.}(2023)\citenamefont {Vanrietvelde}, \citenamefont {Höhn},\ and\ \citenamefont {Giacomini}}]{Vanrietvelde_2018b}%
  \BibitemOpen
  \bibfield  {author} {\bibinfo {author} {\bibfnamefont {A.}~\bibnamefont {Vanrietvelde}}, \bibinfo {author} {\bibfnamefont {P.~A.}\ \bibnamefont {Höhn}},\ and\ \bibinfo {author} {\bibfnamefont {F.}~\bibnamefont {Giacomini}},\ }\bibfield  {title} {\bibinfo {title} {{Switching quantum reference frames in the N-body problem and the absence of global relational perspectives}},\ }\href {https://doi.org/10.22331/q-2023-08-22-1088} {\bibfield  {journal} {\bibinfo  {journal} {Quantum}\ }\textbf {\bibinfo {volume} {7}},\ \bibinfo {pages} {1088} (\bibinfo {year} {2023})}\BibitemShut {NoStop}%
\bibitem [{\citenamefont {Höhn}\ and\ \citenamefont {Vanrietvelde}(2020)}]{Höhn_2018a}%
  \BibitemOpen
  \bibfield  {author} {\bibinfo {author} {\bibfnamefont {P.~A.}\ \bibnamefont {Höhn}}\ and\ \bibinfo {author} {\bibfnamefont {A.}~\bibnamefont {Vanrietvelde}},\ }\bibfield  {title} {\bibinfo {title} {How to switch between relational quantum clocks},\ }\href {https://doi.org/10.1088/1367-2630/abd1ac} {\bibfield  {journal} {\bibinfo  {journal} {New Journal of Physics}\ }\textbf {\bibinfo {volume} {22}},\ \bibinfo {pages} {123048} (\bibinfo {year} {2020})}\BibitemShut {NoStop}%
\bibitem [{\citenamefont {{H{\"o}hn}}(2019)}]{Hoehn_2018b}%
  \BibitemOpen
  \bibfield  {author} {\bibinfo {author} {\bibfnamefont {P.~A.}\ \bibnamefont {{H{\"o}hn}}},\ }\bibfield  {title} {\bibinfo {title} {{Switching Internal Times and a New Perspective on the `Wave Function of the Universe'}},\ }\href {https://doi.org/10.3390/universe5050116} {\bibfield  {journal} {\bibinfo  {journal} {Universe}\ }\textbf {\bibinfo {volume} {5}},\ \bibinfo {pages} {116} (\bibinfo {year} {2019})}\BibitemShut {NoStop}%
\bibitem [{\citenamefont {H\"ohn}\ \emph {et~al.}(2021)\citenamefont {H\"ohn}, \citenamefont {Smith},\ and\ \citenamefont {Lock}}]{Hoehn_2019_trinity}%
  \BibitemOpen
  \bibfield  {author} {\bibinfo {author} {\bibfnamefont {P.~A.}\ \bibnamefont {H\"ohn}}, \bibinfo {author} {\bibfnamefont {A.~R.~H.}\ \bibnamefont {Smith}},\ and\ \bibinfo {author} {\bibfnamefont {M.~P.~E.}\ \bibnamefont {Lock}},\ }\bibfield  {title} {\bibinfo {title} {Trinity of relational quantum dynamics},\ }\href {https://doi.org/10.1103/PhysRevD.104.066001} {\bibfield  {journal} {\bibinfo  {journal} {Phys. Rev. D}\ }\textbf {\bibinfo {volume} {104}},\ \bibinfo {pages} {066001} (\bibinfo {year} {2021})}\BibitemShut {NoStop}%
\bibitem [{\citenamefont {Höhn}\ \emph {et~al.}(2021)\citenamefont {Höhn}, \citenamefont {Smith},\ and\ \citenamefont {Lock}}]{Hoehn_2020}%
  \BibitemOpen
  \bibfield  {author} {\bibinfo {author} {\bibfnamefont {P.~A.}\ \bibnamefont {Höhn}}, \bibinfo {author} {\bibfnamefont {A.~R.~H.}\ \bibnamefont {Smith}},\ and\ \bibinfo {author} {\bibfnamefont {M.~P.~E.}\ \bibnamefont {Lock}},\ }\bibfield  {title} {\bibinfo {title} {{Equivalence of Approaches to Relational Quantum Dynamics in Relativistic Settings}},\ }\href {https://doi.org/10.3389/fphy.2021.587083} {\bibfield  {journal} {\bibinfo  {journal} {Front. in Phys.}\ }\textbf {\bibinfo {volume} {9}},\ \bibinfo {pages} {181} (\bibinfo {year} {2021})},\ \Eprint {https://arxiv.org/abs/2007.00580} {arXiv:2007.00580 [gr-qc]} \BibitemShut {NoStop}%
\bibitem [{\citenamefont {Krumm}\ \emph {et~al.}(2021)\citenamefont {Krumm}, \citenamefont {Höhn},\ and\ \citenamefont {Müller}}]{Krumm_2021}%
  \BibitemOpen
  \bibfield  {author} {\bibinfo {author} {\bibfnamefont {M.}~\bibnamefont {Krumm}}, \bibinfo {author} {\bibfnamefont {P.~A.}\ \bibnamefont {Höhn}},\ and\ \bibinfo {author} {\bibfnamefont {M.~P.}\ \bibnamefont {Müller}},\ }\bibfield  {title} {\bibinfo {title} {Quantum reference frame transformations as symmetries and the paradox of the third particle},\ }\href {https://doi.org/10.22331/q-2021-08-27-530} {\bibfield  {journal} {\bibinfo  {journal} {Quantum}\ }\textbf {\bibinfo {volume} {5}},\ \bibinfo {pages} {530} (\bibinfo {year} {2021})}\BibitemShut {NoStop}%
\bibitem [{\citenamefont {Ahmad}\ \emph {et~al.}(2022)\citenamefont {Ahmad}, \citenamefont {Galley}, \citenamefont {H\"ohn}, \citenamefont {Lock},\ and\ \citenamefont {Smith}}]{hoehn2021quantum}%
  \BibitemOpen
  \bibfield  {author} {\bibinfo {author} {\bibfnamefont {A.~S.}\ \bibnamefont {Ahmad}}, \bibinfo {author} {\bibfnamefont {T.~D.}\ \bibnamefont {Galley}}, \bibinfo {author} {\bibfnamefont {P.~A.}\ \bibnamefont {H\"ohn}}, \bibinfo {author} {\bibfnamefont {M.~P.~E.}\ \bibnamefont {Lock}},\ and\ \bibinfo {author} {\bibfnamefont {A.~R.~H.}\ \bibnamefont {Smith}},\ }\bibfield  {title} {\bibinfo {title} {Quantum relativity of subsystems},\ }\href {https://doi.org/10.1103/PhysRevLett.128.170401} {\bibfield  {journal} {\bibinfo  {journal} {Phys. Rev. Lett.}\ }\textbf {\bibinfo {volume} {128}},\ \bibinfo {pages} {170401} (\bibinfo {year} {2022})}\BibitemShut {NoStop}%
\bibitem [{\citenamefont {Höhn}\ \emph {et~al.}(2022)\citenamefont {Höhn}, \citenamefont {Krumm},\ and\ \citenamefont {Müller}}]{Hoehn_2021}%
  \BibitemOpen
  \bibfield  {author} {\bibinfo {author} {\bibfnamefont {P.~A.}\ \bibnamefont {Höhn}}, \bibinfo {author} {\bibfnamefont {M.}~\bibnamefont {Krumm}},\ and\ \bibinfo {author} {\bibfnamefont {M.~P.}\ \bibnamefont {Müller}},\ }\bibfield  {title} {\bibinfo {title} {{Internal quantum reference frames for finite Abelian groups}},\ }\href {https://doi.org/10.1063/5.0088485} {\bibfield  {journal} {\bibinfo  {journal} {J. Math. Phys.}\ }\textbf {\bibinfo {volume} {63}},\ \bibinfo {pages} {112207} (\bibinfo {year} {2022})},\ \Eprint {https://arxiv.org/abs/2107.07545} {arXiv:2107.07545 [quant-ph]} \BibitemShut {NoStop}%
\bibitem [{\citenamefont {de~la Hamette}\ \emph {et~al.}(2021)\citenamefont {de~la Hamette}, \citenamefont {Galley}, \citenamefont {Höhn}, \citenamefont {Loveridge},\ and\ \citenamefont {Müller}}]{delaHamette_2021_perspectiveneutral}%
  \BibitemOpen
  \bibfield  {author} {\bibinfo {author} {\bibfnamefont {A.-C.}\ \bibnamefont {de~la Hamette}}, \bibinfo {author} {\bibfnamefont {T.~D.}\ \bibnamefont {Galley}}, \bibinfo {author} {\bibfnamefont {P.~A.}\ \bibnamefont {Höhn}}, \bibinfo {author} {\bibfnamefont {L.}~\bibnamefont {Loveridge}},\ and\ \bibinfo {author} {\bibfnamefont {M.~P.}\ \bibnamefont {Müller}},\ }\href@noop {} {\bibinfo {title} {Perspective-neutral approach to quantum frame covariance for general symmetry groups}} (\bibinfo {year} {2021}),\ \Eprint {https://arxiv.org/abs/2110.13824} {arXiv:2110.13824 [quant-ph]} \BibitemShut {NoStop}%
\bibitem [{\citenamefont {de~la Hamette}\ \emph {et~al.}(2022)\citenamefont {de~la Hamette}, \citenamefont {Ludescher},\ and\ \citenamefont {M\"uller}}]{delaHamette_2021_entanglement}%
  \BibitemOpen
  \bibfield  {author} {\bibinfo {author} {\bibfnamefont {A.-C.}\ \bibnamefont {de~la Hamette}}, \bibinfo {author} {\bibfnamefont {S.~L.}\ \bibnamefont {Ludescher}},\ and\ \bibinfo {author} {\bibfnamefont {M.~P.}\ \bibnamefont {M\"uller}},\ }\bibfield  {title} {\bibinfo {title} {{Entanglement-Asymmetry Correspondence for Internal Quantum Reference Frames}},\ }\href {https://doi.org/10.1103/PhysRevLett.129.260404} {\bibfield  {journal} {\bibinfo  {journal} {Phys. Rev. Lett.}\ }\textbf {\bibinfo {volume} {129}},\ \bibinfo {pages} {260404} (\bibinfo {year} {2022})}\BibitemShut {NoStop}%
\bibitem [{\citenamefont {Höhn}\ \emph {et~al.}(2023)\citenamefont {Höhn}, \citenamefont {Kotecha},\ and\ \citenamefont {Mele}}]{Hoehn_2023_subsystems}%
  \BibitemOpen
  \bibfield  {author} {\bibinfo {author} {\bibfnamefont {P.~A.}\ \bibnamefont {Höhn}}, \bibinfo {author} {\bibfnamefont {I.}~\bibnamefont {Kotecha}},\ and\ \bibinfo {author} {\bibfnamefont {F.~M.}\ \bibnamefont {Mele}},\ }\bibfield  {title} {\bibinfo {title} {{Quantum Frame Relativity of Subsystems, Correlations and Thermodynamics}},\ }\href@noop {} {\  (\bibinfo {year} {2023})},\ \Eprint {https://arxiv.org/abs/2308.09131} {arXiv:2308.09131 [quant-ph]} \BibitemShut {NoStop}%
\bibitem [{\citenamefont {Chataignier}\ \emph {et~al.}(2024)\citenamefont {Chataignier}, \citenamefont {Höhn}, \citenamefont {Lock},\ and\ \citenamefont {Mele}}]{Chataignier_2024}%
  \BibitemOpen
  \bibfield  {author} {\bibinfo {author} {\bibfnamefont {L.}~\bibnamefont {Chataignier}}, \bibinfo {author} {\bibfnamefont {P.~A.}\ \bibnamefont {Höhn}}, \bibinfo {author} {\bibfnamefont {M.~P.~E.}\ \bibnamefont {Lock}},\ and\ \bibinfo {author} {\bibfnamefont {F.~M.}\ \bibnamefont {Mele}},\ }\href {https://arxiv.org/abs/2409.06479} {\bibinfo {title} {Relational dynamics with periodic clocks}} (\bibinfo {year} {2024}),\ \Eprint {https://arxiv.org/abs/2409.06479} {arXiv:2409.06479 [quant-ph]} \BibitemShut {NoStop}%
\bibitem [{\citenamefont {Araujo-Regado}\ \emph {et~al.}(2025)\citenamefont {Araujo-Regado}, \citenamefont {Hoehn},\ and\ \citenamefont {Sartini}}]{Araujoregado_2025}%
  \BibitemOpen
  \bibfield  {author} {\bibinfo {author} {\bibfnamefont {G.}~\bibnamefont {Araujo-Regado}}, \bibinfo {author} {\bibfnamefont {P.~A.}\ \bibnamefont {Hoehn}},\ and\ \bibinfo {author} {\bibfnamefont {F.}~\bibnamefont {Sartini}},\ }\href {https://arxiv.org/abs/2506.23459} {\bibinfo {title} {Relational entanglement entropies and quantum reference frames in gauge theories}} (\bibinfo {year} {2025}),\ \Eprint {https://arxiv.org/abs/2506.23459} {arXiv:2506.23459 [hep-th]} \BibitemShut {NoStop}%
\bibitem [{\citenamefont {Loveridge}\ \emph {et~al.}(2018)\citenamefont {Loveridge}, \citenamefont {Miyadera},\ and\ \citenamefont {Busch}}]{Loveridge_2017}%
  \BibitemOpen
  \bibfield  {author} {\bibinfo {author} {\bibfnamefont {L.}~\bibnamefont {Loveridge}}, \bibinfo {author} {\bibfnamefont {T.}~\bibnamefont {Miyadera}},\ and\ \bibinfo {author} {\bibfnamefont {P.}~\bibnamefont {Busch}},\ }\bibfield  {title} {\bibinfo {title} {Symmetry, {Reference} {Frames}, and {Relational} {Quantities} in {Quantum} {Mechanics}},\ }\href {https://doi.org/10.1007/s10701-018-0138-3} {\bibfield  {journal} {\bibinfo  {journal} {Foundations of Physics}\ }\textbf {\bibinfo {volume} {48}},\ \bibinfo {pages} {135--198} (\bibinfo {year} {2018})}\BibitemShut {NoStop}%
\bibitem [{\citenamefont {Carette}\ \emph {et~al.}(2025)\citenamefont {Carette}, \citenamefont {Glowacki},\ and\ \citenamefont {Loveridge}}]{Carette_2023}%
  \BibitemOpen
  \bibfield  {author} {\bibinfo {author} {\bibfnamefont {T.}~\bibnamefont {Carette}}, \bibinfo {author} {\bibfnamefont {J.}~\bibnamefont {Glowacki}},\ and\ \bibinfo {author} {\bibfnamefont {L.}~\bibnamefont {Loveridge}},\ }\bibfield  {title} {\bibinfo {title} {Operational {Q}uantum {R}eference {F}rame {T}ransformations},\ }\href {https://doi.org/10.22331/q-2025-03-27-1680} {\bibfield  {journal} {\bibinfo  {journal} {{Quantum}}\ }\textbf {\bibinfo {volume} {9}},\ \bibinfo {pages} {1680} (\bibinfo {year} {2025})}\BibitemShut {NoStop}%
\bibitem [{\citenamefont {Castro-Ruiz}\ and\ \citenamefont {Oreshkov}(2025)}]{Castro_Ruiz_2021}%
  \BibitemOpen
  \bibfield  {author} {\bibinfo {author} {\bibfnamefont {E.}~\bibnamefont {Castro-Ruiz}}\ and\ \bibinfo {author} {\bibfnamefont {O.}~\bibnamefont {Oreshkov}},\ }\bibfield  {title} {\bibinfo {title} {Relative subsystems and quantum reference frame transformations},\ }\href {https://doi.org/10.1038/s42005-025-02036-x} {\bibfield  {journal} {\bibinfo  {journal} {Communications Physics}\ }\textbf {\bibinfo {volume} {8}},\ \bibinfo {pages} {187} (\bibinfo {year} {2025})}\BibitemShut {NoStop}%
\bibitem [{\citenamefont {Castro-Ruiz}\ \emph {et~al.}(2025)\citenamefont {Castro-Ruiz}, \citenamefont {Galley},\ and\ \citenamefont {Loveridge}}]{Castro-Ruiz_2025}%
  \BibitemOpen
  \bibfield  {author} {\bibinfo {author} {\bibfnamefont {E.}~\bibnamefont {Castro-Ruiz}}, \bibinfo {author} {\bibfnamefont {T.~D.}\ \bibnamefont {Galley}},\ and\ \bibinfo {author} {\bibfnamefont {L.}~\bibnamefont {Loveridge}},\ }\href@noop {} {\bibinfo {title} {{Interpreting quantum reference frame transformations through a simple example}}} (\bibinfo {year} {2025}),\ \Eprint {https://arxiv.org/abs/2508.09540} {arXiv:2508.09540 [quant-ph]} \BibitemShut {NoStop}%
\bibitem [{\citenamefont {Brukner}(2017)}]{Brukner_2017}%
  \BibitemOpen
  \bibfield  {author} {\bibinfo {author} {\bibfnamefont {{\v{C}}.}~\bibnamefont {Brukner}},\ }\bibinfo {title} {On the quantum measurement problem},\ in\ \href {https://doi.org/10.1007/978-3-319-38987-5_5} {\emph {\bibinfo {booktitle} {Quantum [Un]Speakables II: Half a Century of Bell's Theorem}}},\ \bibinfo {editor} {edited by\ \bibinfo {editor} {\bibfnamefont {R.}~\bibnamefont {Bertlmann}}\ and\ \bibinfo {editor} {\bibfnamefont {A.}~\bibnamefont {Zeilinger}}}\ (\bibinfo  {publisher} {Springer International Publishing},\ \bibinfo {address} {Cham},\ \bibinfo {year} {2017})\ pp.\ \bibinfo {pages} {95--117}\BibitemShut {NoStop}%
\bibitem [{\citenamefont {{Frauchiger}}\ and\ \citenamefont {{Renner}}(2018)}]{Frauchiger_2018}%
  \BibitemOpen
  \bibfield  {author} {\bibinfo {author} {\bibfnamefont {D.}~\bibnamefont {{Frauchiger}}}\ and\ \bibinfo {author} {\bibfnamefont {R.}~\bibnamefont {{Renner}}},\ }\bibfield  {title} {\bibinfo {title} {{Quantum theory cannot consistently describe the use of itself}},\ }\href {https://doi.org/10.1038/s41467-018-05739-8} {\bibfield  {journal} {\bibinfo  {journal} {Nature Communications}\ }\textbf {\bibinfo {volume} {9}},\ \bibinfo {eid} {3711} (\bibinfo {year} {2018})}\BibitemShut {NoStop}%
\bibitem [{\citenamefont {Bong}\ \emph {et~al.}(2020)\citenamefont {Bong}, \citenamefont {Utreras-Alarc{\'o}n}, \citenamefont {Ghafari}, \citenamefont {Liang}, \citenamefont {Tischler}, \citenamefont {Cavalcanti}, \citenamefont {Pryde},\ and\ \citenamefont {Wiseman}}]{Bong_2019}%
  \BibitemOpen
  \bibfield  {author} {\bibinfo {author} {\bibfnamefont {K.-W.}\ \bibnamefont {Bong}}, \bibinfo {author} {\bibfnamefont {A.}~\bibnamefont {Utreras-Alarc{\'o}n}}, \bibinfo {author} {\bibfnamefont {F.}~\bibnamefont {Ghafari}}, \bibinfo {author} {\bibfnamefont {Y.-C.}\ \bibnamefont {Liang}}, \bibinfo {author} {\bibfnamefont {N.}~\bibnamefont {Tischler}}, \bibinfo {author} {\bibfnamefont {E.~G.}\ \bibnamefont {Cavalcanti}}, \bibinfo {author} {\bibfnamefont {G.~J.}\ \bibnamefont {Pryde}},\ and\ \bibinfo {author} {\bibfnamefont {H.~M.}\ \bibnamefont {Wiseman}},\ }\bibfield  {title} {\bibinfo {title} {{A strong no-go theorem on the Wigner's friend paradox}},\ }\href {https://doi.org/10.1038/s41567-020-0990-x} {\bibfield  {journal} {\bibinfo  {journal} {Nature Physics}\ }\textbf {\bibinfo {volume} {16}},\ \bibinfo {pages} {1199--1205} (\bibinfo {year} {2020})}\BibitemShut {NoStop}%
\bibitem [{\citenamefont {Vuyst}\ \emph {et~al.}(2025)\citenamefont {Vuyst}, \citenamefont {Hoehn},\ and\ \citenamefont {Tsobanjan}}]{devuyst2025relation}%
  \BibitemOpen
  \bibfield  {author} {\bibinfo {author} {\bibfnamefont {J.~D.}\ \bibnamefont {Vuyst}}, \bibinfo {author} {\bibfnamefont {P.~A.}\ \bibnamefont {Hoehn}},\ and\ \bibinfo {author} {\bibfnamefont {A.}~\bibnamefont {Tsobanjan}},\ }\href {https://arxiv.org/abs/2507.14131} {\bibinfo {title} {On the relation between perspective-neutral, algebraic, and effective quantum reference frames}} (\bibinfo {year} {2025}),\ \Eprint {https://arxiv.org/abs/2507.14131} {arXiv:2507.14131 [quant-ph]} \BibitemShut {NoStop}%
\bibitem [{\citenamefont {Zanardi}(2001)}]{Zanardi_2001}%
  \BibitemOpen
  \bibfield  {author} {\bibinfo {author} {\bibfnamefont {P.}~\bibnamefont {Zanardi}},\ }\bibfield  {title} {\bibinfo {title} {Virtual quantum subsystems},\ }\href {https://doi.org/10.1103/PhysRevLett.87.077901} {\bibfield  {journal} {\bibinfo  {journal} {Phys. Rev. Lett.}\ }\textbf {\bibinfo {volume} {87}},\ \bibinfo {pages} {077901} (\bibinfo {year} {2001})}\BibitemShut {NoStop}%
\bibitem [{\citenamefont {Valente}(2025)}]{Valente2025semester}%
  \BibitemOpen
  \bibfield  {author} {\bibinfo {author} {\bibfnamefont {L.}~\bibnamefont {Valente}},\ }\bibfield  {title} {\bibinfo {title} {Decoherence in quantum reference frames}} (\bibinfo {year} {2025}),\ \bibinfo {note} {semester project, ETH Zurich. Supervised by R.~Renner, F.~Giacomini, and V.~Kabel}\BibitemShut {NoStop}%
\bibitem [{\citenamefont {Henneaux}\ and\ \citenamefont {Teitelboim}(1992)}]{Henneaux_Teitelboim_1992}%
  \BibitemOpen
  \bibfield  {author} {\bibinfo {author} {\bibfnamefont {M.}~\bibnamefont {Henneaux}}\ and\ \bibinfo {author} {\bibfnamefont {C.}~\bibnamefont {Teitelboim}},\ }\href@noop {} {\emph {\bibinfo {title} {Quantization of Gauge Systems}}}\ (\bibinfo  {publisher} {Princeton University Press},\ \bibinfo {address} {Princeton},\ \bibinfo {year} {1992})\BibitemShut {NoStop}%
\bibitem [{\citenamefont {Rovelli}(2014)}]{Rovelli_2013_whygauge}%
  \BibitemOpen
  \bibfield  {author} {\bibinfo {author} {\bibfnamefont {C.}~\bibnamefont {Rovelli}},\ }\bibfield  {title} {\bibinfo {title} {{Why Gauge?}},\ }\href {https://doi.org/10.1007/s10701-013-9768-7} {\bibfield  {journal} {\bibinfo  {journal} {Found. Phys.}\ }\textbf {\bibinfo {volume} {44}},\ \bibinfo {pages} {91--104} (\bibinfo {year} {2014})}\BibitemShut {NoStop}%
\bibitem [{\citenamefont {{Donnelly}}\ and\ \citenamefont {{Freidel}}(2016)}]{Donnelly_2016}%
  \BibitemOpen
  \bibfield  {author} {\bibinfo {author} {\bibfnamefont {W.}~\bibnamefont {{Donnelly}}}\ and\ \bibinfo {author} {\bibfnamefont {L.}~\bibnamefont {{Freidel}}},\ }\bibfield  {title} {\bibinfo {title} {{Local subsystems in gauge theory and gravity}},\ }\href {https://doi.org/10.1007/JHEP09(2016)102} {\bibfield  {journal} {\bibinfo  {journal} {Journal of High Energy Physics}\ }\textbf {\bibinfo {volume} {2016}},\ \bibinfo {eid} {102} (\bibinfo {year} {2016})}\BibitemShut {NoStop}%
\bibitem [{\citenamefont {Gel'fand}\ and\ \citenamefont {Vilenkin}(2016)}]{gel2016generalized}%
  \BibitemOpen
  \bibfield  {author} {\bibinfo {author} {\bibfnamefont {I.~M.}\ \bibnamefont {Gel'fand}}\ and\ \bibinfo {author} {\bibfnamefont {N.~Y.}\ \bibnamefont {Vilenkin}},\ }\href@noop {} {\emph {\bibinfo {title} {Generalized functions, Volume 4: Applications of Harmonic Analysis}}},\ Vol.\ \bibinfo {volume} {380}\ (\bibinfo  {publisher} {American Mathematical Soc.},\ \bibinfo {year} {2016})\BibitemShut {NoStop}%
\bibitem [{\citenamefont {Schwartz}(1950)}]{schwartz1950theorie}%
  \BibitemOpen
  \bibfield  {author} {\bibinfo {author} {\bibfnamefont {L.}~\bibnamefont {Schwartz}},\ }\href@noop {} {\emph {\bibinfo {title} {Th{\'e}orie des distributions: Tome 1. Par L. Schwartz}}}\ (\bibinfo  {publisher} {Hermann (Chartres)},\ \bibinfo {year} {1950})\BibitemShut {NoStop}%
\bibitem [{\citenamefont {Thiemann}(2008)}]{thiemannModernCanonicalQuantum2008}%
  \BibitemOpen
  \bibfield  {author} {\bibinfo {author} {\bibfnamefont {T.}~\bibnamefont {Thiemann}},\ }\href@noop {} {\emph {\bibinfo {title} {Modern Canonical Quantum General Relativity}}}\ (\bibinfo  {publisher} {{Cambridge University Press}},\ \bibinfo {year} {2008})\BibitemShut {NoStop}%
\bibitem [{\citenamefont {Giulini}\ and\ \citenamefont {Marolf}(1999{\natexlab{a}})}]{Giulini:1998kf}%
  \BibitemOpen
  \bibfield  {author} {\bibinfo {author} {\bibfnamefont {D.}~\bibnamefont {Giulini}}\ and\ \bibinfo {author} {\bibfnamefont {D.}~\bibnamefont {Marolf}},\ }\bibfield  {title} {\bibinfo {title} {{A Uniqueness theorem for constraint quantization}},\ }\href {https://doi.org/10.1088/0264-9381/16/7/322} {\bibfield  {journal} {\bibinfo  {journal} {Class. Quantum Grav.}\ }\textbf {\bibinfo {volume} {16}},\ \bibinfo {pages} {2489--2505} (\bibinfo {year} {1999}{\natexlab{a}})}\BibitemShut {NoStop}%
\bibitem [{\citenamefont {Giulini}\ and\ \citenamefont {Marolf}(1999{\natexlab{b}})}]{Giulini:1998rk}%
  \BibitemOpen
  \bibfield  {author} {\bibinfo {author} {\bibfnamefont {D.}~\bibnamefont {Giulini}}\ and\ \bibinfo {author} {\bibfnamefont {D.}~\bibnamefont {Marolf}},\ }\bibfield  {title} {\bibinfo {title} {{On the generality of refined algebraic quantization}},\ }\href {https://doi.org/10.1088/0264-9381/16/7/321} {\bibfield  {journal} {\bibinfo  {journal} {Class. Quantum Grav.}\ }\textbf {\bibinfo {volume} {16}},\ \bibinfo {pages} {2479--2488} (\bibinfo {year} {1999}{\natexlab{b}})}\BibitemShut {NoStop}%
\bibitem [{\citenamefont {Vogel}\ and\ \citenamefont {Risken}(1989)}]{Vogel_1989}%
  \BibitemOpen
  \bibfield  {author} {\bibinfo {author} {\bibfnamefont {K.}~\bibnamefont {Vogel}}\ and\ \bibinfo {author} {\bibfnamefont {H.}~\bibnamefont {Risken}},\ }\bibfield  {title} {\bibinfo {title} {Determination of quasiprobability distributions in terms of probability distributions for the rotated quadrature phase},\ }\href {https://doi.org/10.1103/PhysRevA.40.2847} {\bibfield  {journal} {\bibinfo  {journal} {Phys. Rev. A}\ }\textbf {\bibinfo {volume} {40}},\ \bibinfo {pages} {2847--2849} (\bibinfo {year} {1989})}\BibitemShut {NoStop}%
\bibitem [{\citenamefont {Mekonnen}\ \emph {et~al.}(2025)\citenamefont {Mekonnen}, \citenamefont {Galley},\ and\ \citenamefont {Mueller}}]{Mekonnen_2025}%
  \BibitemOpen
  \bibfield  {author} {\bibinfo {author} {\bibfnamefont {M.}~\bibnamefont {Mekonnen}}, \bibinfo {author} {\bibfnamefont {T.~D.}\ \bibnamefont {Galley}},\ and\ \bibinfo {author} {\bibfnamefont {M.~P.}\ \bibnamefont {Mueller}},\ }\href@noop {} {\bibinfo {title} {{Invariance under quantum permutations rules out parastatistics}}} (\bibinfo {year} {2025}),\ \Eprint {https://arxiv.org/abs/2502.17576} {arXiv:2502.17576 [quant-ph]} \BibitemShut {NoStop}%
\bibitem [{\citenamefont {de~la Hamette}\ \emph {et~al.}(2023)\citenamefont {de~la Hamette}, \citenamefont {Kabel}, \citenamefont {Castro-Ruiz},\ and\ \citenamefont {Brukner}}]{delaHamette_2021_indefinitemetric}%
  \BibitemOpen
  \bibfield  {author} {\bibinfo {author} {\bibfnamefont {A.-C.}\ \bibnamefont {de~la Hamette}}, \bibinfo {author} {\bibfnamefont {V.}~\bibnamefont {Kabel}}, \bibinfo {author} {\bibfnamefont {E.}~\bibnamefont {Castro-Ruiz}},\ and\ \bibinfo {author} {\bibfnamefont {{\v C}.}~\bibnamefont {Brukner}},\ }\bibfield  {title} {\bibinfo {title} {Quantum reference frames for an indefinite metric},\ }\href {https://doi.org/10.1038/s42005-023-01344-4} {\bibfield  {journal} {\bibinfo  {journal} {Communications Physics}\ }\textbf {\bibinfo {volume} {6}},\ \bibinfo {pages} {231} (\bibinfo {year} {2023})}\BibitemShut {NoStop}%
\bibitem [{\citenamefont {{Kitaev}}\ \emph {et~al.}(2004)\citenamefont {{Kitaev}}, \citenamefont {{Mayers}},\ and\ \citenamefont {{Preskill}}}]{Kitaev_2004}%
  \BibitemOpen
  \bibfield  {author} {\bibinfo {author} {\bibfnamefont {A.}~\bibnamefont {{Kitaev}}}, \bibinfo {author} {\bibfnamefont {D.}~\bibnamefont {{Mayers}}},\ and\ \bibinfo {author} {\bibfnamefont {J.}~\bibnamefont {{Preskill}}},\ }\bibfield  {title} {\bibinfo {title} {{Superselection rules and quantum protocols}},\ }\href {https://doi.org/10.1103/PhysRevA.69.052326} {\bibfield  {journal} {\bibinfo  {journal} {Physical Review A}\ }\textbf {\bibinfo {volume} {69}},\ \bibinfo {eid} {052326} (\bibinfo {year} {2004})}\BibitemShut {NoStop}%
\bibitem [{\citenamefont {Bartlett}\ \emph {et~al.}(2007)\citenamefont {Bartlett}, \citenamefont {Rudolph},\ and\ \citenamefont {Spekkens}}]{Bartlett_2007}%
  \BibitemOpen
  \bibfield  {author} {\bibinfo {author} {\bibfnamefont {S.~D.}\ \bibnamefont {Bartlett}}, \bibinfo {author} {\bibfnamefont {T.}~\bibnamefont {Rudolph}},\ and\ \bibinfo {author} {\bibfnamefont {R.~W.}\ \bibnamefont {Spekkens}},\ }\bibfield  {title} {\bibinfo {title} {Reference frames, superselection rules, and quantum information},\ }\href {https://doi.org/10.1103/revmodphys.79.555} {\bibfield  {journal} {\bibinfo  {journal} {Reviews of Modern Physics}\ }\textbf {\bibinfo {volume} {79}},\ \bibinfo {pages} {555–609} (\bibinfo {year} {2007})}\BibitemShut {NoStop}%
\bibitem [{\citenamefont {{Palmer}}\ \emph {et~al.}(2014)\citenamefont {{Palmer}}, \citenamefont {{Girelli}},\ and\ \citenamefont {{Bartlett}}}]{Palmer_2014}%
  \BibitemOpen
  \bibfield  {author} {\bibinfo {author} {\bibfnamefont {M.~C.}\ \bibnamefont {{Palmer}}}, \bibinfo {author} {\bibfnamefont {F.}~\bibnamefont {{Girelli}}},\ and\ \bibinfo {author} {\bibfnamefont {S.~D.}\ \bibnamefont {{Bartlett}}},\ }\bibfield  {title} {\bibinfo {title} {{Changing quantum reference frames}},\ }\href {https://doi.org/10.1103/PhysRevA.89.052121} {\bibfield  {journal} {\bibinfo  {journal} {Physical Review A}\ }\textbf {\bibinfo {volume} {89}},\ \bibinfo {eid} {052121} (\bibinfo {year} {2014})}\BibitemShut {NoStop}%
\bibitem [{\citenamefont {Riello}(2021)}]{Riello_2021}%
  \BibitemOpen
  \bibfield  {author} {\bibinfo {author} {\bibfnamefont {A.}~\bibnamefont {Riello}},\ }\href@noop {} {\bibinfo {title} {{Edge modes without edge modes}}} (\bibinfo {year} {2021}),\ \Eprint {https://arxiv.org/abs/2104.10182} {arXiv:2104.10182 [hep-th]} \BibitemShut {NoStop}%
\bibitem [{\citenamefont {Santo}\ \emph {et~al.}(2025)\citenamefont {Santo}, \citenamefont {Manzano},\ and\ \citenamefont {Brukner}}]{delsanto2025}%
  \BibitemOpen
  \bibfield  {author} {\bibinfo {author} {\bibfnamefont {F.~D.}\ \bibnamefont {Santo}}, \bibinfo {author} {\bibfnamefont {G.}~\bibnamefont {Manzano}},\ and\ \bibinfo {author} {\bibfnamefont {C.}~\bibnamefont {Brukner}},\ }\href {https://arxiv.org/abs/2407.06279} {\bibinfo {title} {Wigner's friend scenarios: on what to condition and how to verify the predictions}} (\bibinfo {year} {2025}),\ \Eprint {https://arxiv.org/abs/2407.06279} {arXiv:2407.06279 [quant-ph]} \BibitemShut {NoStop}%
\bibitem [{\citenamefont {Wigner}(1995)}]{Wigner_1962}%
  \BibitemOpen
  \bibfield  {author} {\bibinfo {author} {\bibfnamefont {E.~P.}\ \bibnamefont {Wigner}},\ }\bibinfo {title} {Remarks on the mind-body question},\ in\ \href {https://doi.org/10.1007/978-3-642-78374-6_20} {\emph {\bibinfo {booktitle} {Philosophical Reflections and Syntheses}}},\ \bibinfo {editor} {edited by\ \bibinfo {editor} {\bibfnamefont {J.}~\bibnamefont {Mehra}}}\ (\bibinfo  {publisher} {Springer Berlin Heidelberg},\ \bibinfo {address} {Berlin, Heidelberg},\ \bibinfo {year} {1995})\ pp.\ \bibinfo {pages} {247--260}\BibitemShut {NoStop}%
\bibitem [{\citenamefont {Nielsen}(2000)}]{nielsen2000quantuminformationtheory}%
  \BibitemOpen
  \bibfield  {author} {\bibinfo {author} {\bibfnamefont {M.~A.}\ \bibnamefont {Nielsen}},\ }\href {https://arxiv.org/abs/quant-ph/0011036} {\bibinfo {title} {Quantum information theory}} (\bibinfo {year} {2000}),\ \Eprint {https://arxiv.org/abs/quant-ph/0011036} {arXiv:quant-ph/0011036 [quant-ph]} \BibitemShut {NoStop}%
\bibitem [{\citenamefont {{Tyson}}(2003)}]{Tyson_2003}%
  \BibitemOpen
  \bibfield  {author} {\bibinfo {author} {\bibfnamefont {J.~E.}\ \bibnamefont {{Tyson}}},\ }\bibfield  {title} {\bibinfo {title} {{Operator-Schmidt decompositions and the Fourier transform, with applications to the operator-Schmidt numbers of unitaries}},\ }\href {https://doi.org/10.1088/0305-4470/36/39/309} {\bibfield  {journal} {\bibinfo  {journal} {Journal of Physics A Mathematical General}\ }\textbf {\bibinfo {volume} {36}},\ \bibinfo {pages} {10101--10114} (\bibinfo {year} {2003})},\ \Eprint {https://arxiv.org/abs/quant-ph/0306144} {arXiv:quant-ph/0306144 [quant-ph]} \BibitemShut {NoStop}%
\end{thebibliography}%
\bibliographystyle{apsrev4-2.bst}

\appendix
\section{Unitarity and Transitivity of the Extended QRF Transformation} \label{app:genQRFtrafosPersp}

Let us start by checking that the generalised QRF change operator in the perspectival approach is transitive: $\hat{S}_{B\to C}^P\hat{S}_{A\to B}^P = \hat{S}_{A\to C}^P$. We show this on a general state of $B$ and $C$ relative to $A$:
\begin{align}
    \hat{S}_{B\to C}^P\hat{S}_{A\to B}^P \ket{\psi^P}_{BC}^{(A)} 
    &= \hat{S}_{B\to C}^P  \hat{\mathcal{P}}_{AB} e^{i\hat{x}_B(\hat{p}_C-P)}\int dy dz\ \psi(y,z)\ket{y}_B\ket{z}_C\\
    &= \hat{S}_{B\to C}^P \hat{\mathcal{P}}_{AB}\int dy dz \ \psi(y,z) e^{-iPy}  \ket{y}_B\ket{z-y}_C \\ 
    &= \hat{S}_{B\to C}^P  \int dy dz\ \psi(y,z) e^{-iPy} \ket{-y}_A \ket{z-y}_C\\
    &=  \hat{\mathcal{P}}_{BC}\int dy dz\  \psi(y,z) e^{i(-y-(z-y))P} \ket{-y-(z-y)}_A \ket{z-y}_C \\
    &= \int dy dz\  \psi(y,z)  e^{-izP}\ket{-z}_A\ket{y-z}_B \\
    &= \hat{S}_{A\to C}^P\ket{\psi^P}_{BC}^{(A)}. 
\end{align}

Moreover, we can show that the generalised QRF change operator is unitary: $(\hat{S}_{A\to B}^P)^\dagger=\hat{S}_{B\to A}^P$.
\begin{align}
    (\hat{S}_{A\to B}^P)^\dagger = (\hat{\mathcal{P}}_{AB} e^{i\hat{x}_B(\hat{p}_C-P)})^\dagger = e^{-i\hat{x}_B(\hat{p}_C-P)}\hat{\mathcal{P}}_{AB}^\dagger = \hat{\mathcal{P}}_{AB}^\dagger e^{i\hat{x}_A(\hat{p}_C-P)} = \hat{\mathcal{P}}_{BA}e^{i\hat{x}_A(\hat{p}_C-P)} = \hat{S}_{B\to A}^P.
\end{align}

\section{Extended QRF Transformation in the Perspective-Neutral Approach} \label{app:genQRFtrafosPN}

In this Appendix we show that the generalised QRF transformation operator in the perspective-neutral approach, provided in Eq.~\eqref{eq:genQRFtrafoPN} indeed provides the right mapping between relative physical states. For convenience, let us repeat that the transformation operator from $A$ at position $X$ to $C$ at position $Z$ in charge sector $P$ is given by
\begin{align}
    \hat{S}^P_{A\to C}(X,Z)=\int dz \ket{X+Z-z}_A\bra{z}_C\otimes \hat{U}_B^P(Z-z))=\int dz \ket{X+Z-z}_A\bra{z}_C\otimes e^{i(z-Z)(\hat{p}_B-P)}.\label{eq:genQRFchangeXZ}
\end{align}
Let us now show that this operator indeed maps a general relative physical state relative to $A$ of the form 
\begin{align}
    \ket{\psi^P}_{BC}^{(A)}= \int dx dy dz\ e^{-iP(x-X)}\psi(x,y,z) \ket{y-(x-X)}_B \ket{z-(x-X)}_C\label{eq:GenRelStateA}
\end{align}
to the corresponding state relative to $C$
\begin{align}
    \ket{\psi^P}_{AB}^{(C)}= \int dx dy dz\ e^{-iP(z-Z)}\psi(x,y,z) \ket{x-(z-Z)}_A \ket{y-(z-Z)}_B.\label{eq:genRelStateC}
\end{align}
(these are the generalisations of Eq.~\eqref{eq:psiRelAX} and Eq.~\eqref{eq:psiRelCZ}). We can calculate this straightforwardly:
\begin{align}
   &\hat{S}_{A\to C}^P(X,Z) \ket{\psi^P}_{BC}^{(A)}\\ &=  \int dz' \ket{X+Z-z'}_A\bra{z'}_C\otimes e^{i(z'-Z)(\hat{p}_B-P)} \int dx dy dz \ e^{-iP(x-X)}\psi(x,y,z) \ket{y-(x-X)}_B \ket{z-(x-X)}_C \\
   &= \int dx dy dz\ e^{i(z-Z-(x-X))(\hat{p}_B-P)}e^{-iP(x-X)}\psi(x,y,z) \ket{X+Z-(z-(x-X))}_A \ket{y-(x-X)}_B \\ 
   &= \int dx dy dz\ e^{-iP(x-X)}e^{-iP(z-Z-(x-X))}\psi(x,y,z) \ket{x-(z-Z)}_A \ket{y-(x-X)-(z-Z)+(x-X)}_B \\
   &= \int dx dy dz\ e^{-iP(z-Z)}\psi(x,y,z) \ket{x-(z-Z)}_A \ket{y-(z-Z)}_B \\
   &= \ket{\psi^P}_{AB}^{(C)}.
\end{align}
It is straightforward to check that the operator $\hat{S}_{A\to C}^P(X,Z)$ is unitary and transitive as well.

Moreover, if we set $X=Y=0$, we recover the generalised QRF transformation operator from the perspectival approach:
\begin{align}
    \hat{S}_{A\to C}^P(0,0) =\int dz \ket{X+Z-z}_A\bra{z}_C\otimes \hat{U}_B^P(Z-z))=\int dz \ket{-z}_A\bra{z}_C\otimes e^{iz(\hat{p}_B-P)}=\hat{\mathcal{P}}_{AC}e^{i\hat{x}_C(\hat{p}_B-P)},
\end{align}
thus showing consistency between the two approaches when the reference frames are considered to be at the origin relative to themselves.

\section{Invariance of Observables under Extended QRF Transformation}\label{app: invariantOperators}

\setcounter{theorem}{0}
\begin{theorem}
Given a Hermitian operator with Schmidt decomposition $\hat{O}_{BC}^{(A)}=\sum_{k=1}^{n}\lambda_k \hat{O}^{(k)}_B\otimes \hat{O}^{(k)}_C\in \mathcal{L}(\mathcal{H}_{B}^{(A)}\otimes \mathcal{H}_{C}^{(A)})$, which is invariant under the QRF transformation from $A$ to $C$ for zero charge sector, i.e.~$\hat{S}_{A\to C} \hat{O}_{BC}^{(A)} \hat{S}^\dagger_{A\to C} =\mathcal{P}_{AC} \hat{O}_{BC}^{(A)}\mathcal{P}_{AC}^\dagger$, it is invariant under the generalised QRF transformation $\hat{S}^P_{A\to C}$ from $A$ to $C$ if and only if $[\hat{x}_C,\hat{O}_C^{(k)}]=0$ for all $k\in\{1,\dots,n\}$.
\end{theorem}
\begin{proof}
We assume that an operator $\hat{O}_{BC}^{(A)}$ is invariant under the standard QRF transformation if
\begin{equation}
    \hat{S}_{A\to C} \hat{O}_{BC}^{(A)} \hat{S}^\dagger_{A\to C} = \mathcal{P}_{AC} e^{i\hat{x}_C\hat{p}_B}\hat{O}_{BC}^{(A)}e^{-i\hat{x}_C\hat{p}_B}\mathcal{P}_{AC}^\dagger =\mathcal{P}_{AC} \hat{O}_{BC}^{(A)}\mathcal{P}_{AC}^\dagger,\label{eq: invariance}
\end{equation}
i.e.~it remains the same up to a parity swap of $A$ and $C$. Note that for self-adjoint operators, $e^{iX}Ye^{-iX} = Y \Leftrightarrow [X,Y]=0$.\footnote{For unbounded operators like $\hat{p}_B$ and $\hat{x}_C$, all operator equations are understood on appropriate dense domains where the operators and their products are well-defined.} Thus, the invariance condition in Eq.~\eqref{eq: invariance} is equivalent to $[\hat{p}_B\otimes\hat{x}_C,\hat{O}^{(A)}_{BC}]=0$.\\

Similarly, the condition for the operator to be invariant under the extended QRF transformation reads
\begin{equation}
    \hat{S}_{A\to C} \hat{O}_{BC}^{(A)} \hat{S}^\dagger_{A\to C} = \mathcal{P}_{AC} e^{i\hat{x}_C\hat{p}_B}e^{-i\hat{x}_CP}\hat{O}_{BC}^{(A)}e^{i\hat{x}_CP}e^{-i\hat{x}_C\hat{p}_B}\mathcal{P}_{AC}^\dagger =\mathcal{P}_{AC} \hat{O}_{BC}^{(A)}\mathcal{P}_{AC}^\dagger.
\end{equation}
If the operator is already invariant under the standard QRF transformation, this is true if and only if additionally
\begin{equation}
    e^{-i\hat{x}_CP}\hat{O}_{BC}^{(A)}e^{i\hat{x}_CP}=\hat{O}_{BC}^{(A)}.
\end{equation}
Again, this is equivalent to the condition $[\mathbb{I}_B\otimes\hat{x}_C,\hat{O}_{BC}^{(A)}]=0$. Using the orthogonality of the Schmidt decomposition, this imposes a condition on each $\hat{O}_C^{(k)}$ separately. To see this, note first that we can rewrite the commutation condition as
\begin{equation}
    [\mathbb{I}_B\otimes \hat{x}_C,\hat{O}_{BC}^{(A)}]= \sum_{k=1}^n [\mathbb{I}_B\otimes \hat{x}_C,O_B^{(k)}\otimes O_C^{(k)}]= \sum_{k=1}^n\hat{O}_B^{(k)}\otimes [\hat{x}_C,\hat{O}_C^{(k)}]=0.
\end{equation}
Now, in the Schmidt decomposition, the operators in each term must be orthonormal, that is $\mathrm{Tr}\left((\hat{O}_X^{(k)})^\dagger \hat{O}_X^{(j)}\right)=\delta_{kj}$ for $X\in \{B,C\}$ and $j,k\in \{1,\dots,n\}$ \cite{nielsen2000quantuminformationtheory, Tyson_2003}. This implies that all terms in the sum are linearly independent and thus the condition is satisfied if and only if $[\hat{x}_C,\hat{O}_C^{(k)}]=0$ for all $k\in{1,\dots,n}$.
\end{proof}

In order to gain some intuition for this condition, it is helpful to study more closely the conditions under which an operator is invariant under the standard QRF transformation but not under the extended QRF transformation for non-zero charge sector. To do so, we consider as an example an observable for which the Schmidt rank $n=1$, such that it can be written as $\hat{O}_{BC}^{(A)}=\hat{O}_B\otimes \hat{O}_C$. In this case the following proposition holds.
\setcounter{proposition}{0}
\begin{proposition}
Given a non-zero Hermitian operator of the form $\hat{O}_{BC}^{(A)}=\hat{O}_B\otimes \hat{O}_C\in \mathcal{H}_{BC}^{(A)}$ that is invariant under the QRF transformation from $A$ to $C$ for zero charge sector, this invariance does \emph{not} extend to the transformation for non-zero charge sector if and only if $\{\hat{p}_B,\hat{O}_B\}=0$ and $\{\hat{x}_C,\hat{O}_C\}=0$.
\end{proposition}
\begin{proof}
The condition for invariance under the standard QRF transformation $\hat{S}_{A\to C}$,
\begin{equation}
[\hat{p}_B\otimes \hat{x}_C,\hat{O}^{(A)}_{BC}]=[\hat{p}_B,\hat{O}_B]\otimes\{\hat{p}_C,\hat{O}_C\}+\{\hat{p}_B,\hat{O}_B\}\otimes[\hat{x}_C,\hat{O}_C] = A \otimes D + B \otimes C = 0.
\end{equation}
admits only two non-trivial solutions:
\begin{description}
\item[Case (1): Commuting case] $A = 0$ and $C = 0$ (i.e., $[\hat{p}_B,\hat{O}_B]=0$ and $[\hat{x}_C,\hat{O}_C]=0$)
\item[Case (2): Anticommuting case] $B = 0$ and $D = 0$ (i.e., $\{\hat{p}_B,\hat{O}_B\}=0$ and $\{\hat{x}_C,\hat{O}_C\}=0$)
\end{description}
While there appears to be a third option where $A \otimes D = -B \otimes C$, this would require
\begin{equation}
    [\hat{p}_B,\hat{O}_B]\otimes\{\hat{x}_C,\hat{O}_C\}=-\{\hat{p}_B,\hat{O}_B\}\otimes[\hat{x}_C,\hat{O}_C],\label{eq: option3}
\end{equation}
implying the existence of $\alpha\in\mathbb{R}$ such that $[\hat{p}_B,\hat{O}_B]=\alpha\{\hat{p}_B,\hat{O}_B\}$ and $\{\hat{x}_C,\hat{O}_C\} = -\alpha^{-1}[\hat{x}_C,\hat{O}_C]$. However, since commutators are skew-Hermitian and anti-commutators are Hermitian, they cannot be proportional unless both vanish. This would make $\hat{O}_B=0$ or $\hat{O}_C=0$, yielding a trivial observable. 

Case (1) directly implies that $[\hat{x}_C,\hat{O}_C] =0$, so in this case, the observable is also invariant under the extended QRF transformation. Case (2), on the other hand, characterises precisely the observables which are invariant under the QRF transformation for $P=0$ but not $P\neq 0$ -- after all, a non-zero observable cannot both commute and anti-commute with $\hat{x}_C$.
\end{proof}

\section{A Purely Perspectival Game} \label{app:perspectivalGame} 

Given that our accessibility  game in the main text involves only measurements from, and communication between, perspectives, one may wonder whether an external state prepared by Eve using an external reference frame is necessary for the task. Indeed, we can describe a purely perspectival variant of the accessibility game that does not require an external referee Eve, with the aim of testing whether the perspectival states of Alice and Charlie are compatible with the global charge. There, Alice takes on the role of state preparer, directly preparing Charlie’s reference frame \(C\) and system \(B\) in states of the form \(\ket{\psi}_{BC}^{(A)}\) across different rounds.

From Charlie’s perspective, this preparation manifests as a state on systems \(A\) and \(B\), since Alice’s frame \(A\) is now part of the composite system Charlie observes. Charlie then performs complete tomography on this joint state \(\ket{\psi}_{AB}^{(C)}\). Starting from this situation, they can play any of the variants of the game from the main text. For example, Alice may prepare a state (cf.~Eq.~\eqref{eq:staterelAlice})
\begin{align}  \ket{\psi}^{(A)}_{BC}&= e^{-ix_1P}\frac{1}{\sqrt{2}}\left(\ket{\phi^{(1)}}^{(A)}_{BC}+e^{i\Phi^{(A)}}\ket{\phi^{(2)}}^{(A)}_{BC}\right)
\end{align}
which appears from Charlie’s perspective as (cf.~Eq.~\eqref{eq:staterelCharlie})
\begin{align}
\ket{\psi}^{(C)}_{AB}&= e^{-iz_1P}\frac{1}{\sqrt{2}}\left(\ket{\tilde{\phi}^{(1)}}^{(C)}_{AB}+e^{i\Phi^{(C)}}\ket{\tilde{\phi}^{(2)}}^{(C)}_{AB}\right).
\end{align}
We can assume that Alice knows the state and the phase $\Phi^{(A)}\equiv \varphi + (x_1-x_2)P$ (cf.~Eq.~\eqref{eq: phase relative to A}) that she prepares, and that Charlie can perform a measurement revealing the phase $\Phi^{(C)}\equiv \varphi + (z_1-z_2)P$ (cf.~Eq.~\eqref{eq: phase relative to C}). Upon communication, they can determine \(P\). Equally well, Alice can prepare an arbitrary state of the form
(cf.~Eq.~\eqref{eq:GenRelStateA})
\begin{equation}
\ket{\psi}_{BC}^{(A)} = \int dx dy dz e^{-iPx}\psi(x,y,z) \ket{y-x}_B \ket{z-x}_C ,
\end{equation}
which from Charlie’s perspective looks like (cf.~Eq.~\eqref{eq:genRelStateC})
\begin{equation}
\ket{\psi}_{AB}^{(C)} = \int dx dy dz e^{-iPz}\psi(x,y,z) \ket{x-z}_A \ket{y-z}_B ,
\end{equation}
and which he can determine tomographically. Again, upon communication of the relevant data, they can identify the charge sector \(P\).

Finally, Alice can repeat her preparation procedure across multiple rounds, each time selecting a different state from a collection of linearly independent states within a single charge sector. After collecting measurement data from sufficiently many rounds, Alice and Charlie pool their information. From Alice’s perspective, she knows which states she prepared; from Charlie’s perspective, he has tomographically reconstructed the corresponding output states. Together, these input–output pairs allow them to reconstruct the quantum reference-frame transformation \(\hat{S}_{A \to C}^P\) that relates their perspectives. Since the QRF transformation within a single charge sector is uniquely determined by the total momentum \(P\), reconstructing this unitary directly reveals the charge-sector value. The precision of this reconstruction improves as Alice prepares more linearly independent probe states.\\

These formulations of the game demonstrate that internal observers can test whether their perspectives are compatible with a definite charge -- such as total momentum -- and, if so, which one, via a purely relational protocol, without requiring the existence of an external perspective with an external reference frame or predetermined state preparation by a third party.

\section{Most General Form of a QRF Transformation} \label{app:mostgeneralQRFtrafo}

As noted above, the most general QRF transformation that implements a quantum-controlled translation could include a branch-dependent phase factor that could, in principle, depend on all phase space variables of all systems as seen from the initial reference frame:
\begin{align}
    \hat{S}^{\tilde{f}}_{A \to C}=\hat{\mathcal{P}}_{AC} e^{i\hat{x}_C\hat{p}_B} \tilde{f}(\hat{x}_C,\hat{p}_C,\hat{x}_B,\hat{p}_B).
\end{align}

First, in order to ensure that $\hat{S}_{A \to C}$ is a unitary operator that preserves the norm of states, we must be able to write it instead as
\begin{align}
    \hat{S}^g_{A \to C}=\hat{\mathcal{P}}_{AC} e^{i\hat{x}_C\hat{p}_B} e^{ig(\hat{x}_C,\hat{p}_C,\hat{x}_B,\hat{p}_B)}.
\end{align}

Furthermore, $g$ cannot depend on the momenta of systems $C$ and $B$, since this would induce additional translations of these systems and thus prevent the transformation from implementing a quantum-controlled translation by $\hat{x}_C$. Nevertheless, the above operator could still take the form
\begin{align}
    \hat{S}^f_{A \to C}=\hat{\mathcal{P}}_{AC} e^{i\hat{x}_C\hat{p}_B} e^{if(\hat{x}_C,\hat{x}_B)}.
\end{align}

By requiring unitarity (and assuming $f$ is a real-valued function) we find the following condition:
\begin{align}
    (\hat{S}^f_{A \to C})^\dagger = (\hat{\mathcal{P}}_{AC} e^{i\hat{x}_C\hat{p}_B} e^{if(\hat{x}_C,\hat{x}_B)})^\dagger &= e^{-if(\hat{x}_C,\hat{x}_B)}e^{-i\hat{x}_C\hat{p}_B}\hat{\mathcal{P}}_{AC}^\dagger \\
    &=\hat{\mathcal{P}}_{CA} e^{i\hat{x}_A\hat{p}_B}e^{-if(-\hat{x}_A,\hat{x}_B)}+\left( e^{-if(-\hat{x}_A,\hat{x}_B)} - e^{-if(-\hat{x}_A,\hat{x}_B+\hat{x}_A)} \right)e^{i\hat{x}_A\hat{p}_B} \\
    &\stackrel{!}{=} \hat{\mathcal{P}}_{CA} e^{i\hat{x}_A\hat{p}_B} e^{if(\hat{x}_A,\hat{x}_B)} = \hat{S}^f_{C \to A}.
\end{align}
Thus, (i) $f$ needs to be odd in its first argument and (ii) $f(x,y-x)-f(x,y)\in 2\pi \mathbb{Z}$.

Assume $f:\mathbb{R}^2\to\mathbb{R}$ is continuous in its second argument and satisfies condition (ii),
\begin{equation}
    f(x,y-x)-f(x,y)\in 2\pi\mathbb{Z}\qquad(\forall\,x,y\in\mathbb{R}).
    \label{eq:integer-shift}
\end{equation}
Fix $x$ and define $h(y):=f(x,y)$ as well as
\begin{equation}
    \Delta_x(y):=h(y-x)-h(y).
\end{equation}
By assumption, $\Delta_x(y)\in 2\pi\mathbb{Z}$ for all $y$. Since $\Delta_x$ is continuous in $y$ but takes values in the discrete set $2\pi\mathbb{Z}$, it must be constant:
\begin{equation}
    \Delta_x(y)\equiv 2\pi m(x),\qquad m(x)\in\mathbb{Z}.
\end{equation}
Hence $h$ obeys
\begin{equation}
    h(y-x)-h(y)=2\pi m(x).
    \label{eq:difference-eq}
\end{equation}
We first show that $m(x)$ is constant in $x$ and vanishes. To see that it is constant, note that $h(y-x_1-x_2) -h(y) = 2\pi m(x_1 + x_2)$ is equivalent to
\begin{align}
    \underset{=2\pi m(x_1)}{\underbrace{h(y-x_1-x_2) - h(y-x_2)}} + \underset{=2\pi m(x_2)}{\underbrace{h(y-x_2) -h(y)}} &= 2\pi m(x_1+x_2),\\\nonumber
\end{align}
and thus $m(x_1) + m(x_2) = m(x_1+x_2)$. One can also show that $m(qx) = qm(x)$ for rational numbers $q$, using the linearity multiple times. In particular, this means that $m(\frac{x}{n}) = \frac{1}{n}m(x)$, which still needs to be an integer. But in order to be an integer for all $n$, $m(x)$ must be zero for all $x$.

This then implies that $h(y-x) = h(y)$, which implies that it must be a constant function (unaffected by shifts). Thus, $f(x,y)$ is constant in the second variable $y$ and therefore only a function of $x$.
We conclude that the most general form that the QRF operator for quantum-controlled translations can take is
\begin{align}
    \hat{S}^f_{A \to C}=\hat{\mathcal{P}}_{AC} e^{i\hat{x}_C\hat{p}_B} e^{if(\hat{x}_C)}.
\end{align}

\end{document}